\documentclass[sigconf,authorversion,nonacm]{acmart}

\usepackage{booktabs} 
\usepackage[ruled, linesnumbered, vlined]{algorithm2e}
\usepackage[noend]{algorithmic}

\SetKwInOut{Input}{Input}
\SetKwInOut{Output}{Output}
\SetKwInOut{Variables}{Variables}
\SetKwComment{comm}{\hfill$\triangleright$\ }{}

\usepackage{hyperref}
\usepackage{balance}

\usepackage{color}
\usepackage{multirow}
\usepackage{subfigure}
\usepackage{enumerate}

\usepackage{xspace}

\newtheorem{theorem}{Theorem}

\newtheorem{lemma}{Lemma}

\newcommand{\todo}[1]{[\textcolor{red}{TODO: #1}]}

\newcommand{\hide}[1]{}
\newcommand{\eat}[1]{}
\newcommand{\stitle}[1]{\vspace{1ex}\noindent\textbf{#1}}

\newcommand{\HIDX}{\ensuremath{\mathcal{H}}\xspace}

\newcommand{\rev}[1]{#1}

\AtBeginDocument{%
  \providecommand\BibTeX{{%
    \normalfont B\kern-0.5em{\scshape i\kern-0.25em b}\kern-0.8em\TeX}}}

\eat{
\setcopyright{acmcopyright}
\copyrightyear{2022}
\acmYear{2022}

\acmConference[SIGMOD '22]{SIGMOD '22: ACM Symposium on Neural
  Gaze Detection}{June 12--17, 2022}{Philadelphia, PA, USA}
\acmBooktitle{SIGMOD '22: ACM SIGMOD International Conference on Management of Data,
  June 12--17, 2022, Philadelphia, PA, USA}
\acmPrice{15.00}
\acmISBN{978-1-4503-XXXX-X/18/06}
}

\begin{document}

\title{HINT: A Hierarchical Index for Intervals in Main Memory}

 \author{George Christodoulou}
 \affiliation{%
   \institution{University of Ioannina}
   \country{Greece}
 }
 \email{gchristodoulou@cse.uoi.gr}

 \author{Panagiotis Bouros}
 \orcid{0000-0002-8846-4330}
 \affiliation{%
   \institution{Johannes Gutenberg University Mainz}
   \country{Germany}
 }
 \email{bouros@uni-mainz.de}

 \author{Nikos Mamoulis}
 \orcid{0000-0003-3423-4895}
 \affiliation{%
   \institution{University of Ioannina}
   \country{Greece}
 }
 \email{nikos@cse.uoi.gr}

\renewcommand{\shortauthors}{Christodoulou et al.}

\begin{abstract}
Indexing intervals is a fundamental problem, finding
a wide range of applications, most notably in temporal and uncertain databases.
In this paper, we propose HINT, a novel and efficient in-memory index for
intervals, with a focus on interval overlap queries, which are a
basic component of many search and analysis tasks.
HINT applies a hierarchical partitioning approach, which assigns each interval to at most two partitions per level and has controlled space requirements.
We reduce the information stored
at each partition to the absolutely necessary by dividing the
intervals in it based on whether they begin inside
or before the partition boundaries.
In addition, our index includes storage optimization techniques for
the effective handling of data sparsity and skewness.
Experimental results on real and synthetic interval sets of different
characteristics show that HINT is typically one order of
magnitude faster than existing interval indexing methods.
\end{abstract}

\eat{

}


\eat{
\begin{teaserfigure}
  \includegraphics[width=\textwidth]{sampleteaser}
  \caption{Seattle Mariners at Spring Training, 2010.}
  \Description{Enjoying the baseball game from the third-base
  seats. Ichiro Suzuki preparing to bat.}
  \label{fig:teaser}
\end{teaserfigure}
}

\maketitle

\section{Introduction}
\label{sec:intro}
A wide range of applications require managing large collections of intervals. In temporal databases \cite{SnodgrassA86,BohlenDGJ17}, each tuple has a {\em validity interval}, which captures the period of time that the tuple is valid. In statistics and probabilistic databases \cite{DalviS04}, uncertain values are often approximated by (confidence or uncertainty) intervals.
In data anonymization \cite{SamaratiS98}, attribute values are often
generalized to value ranges. XML data indexing techniques
\cite{MinPC03} encode label paths as intervals and evaluate path
expressions using containment relationships between the
intervals. Several computational geometry problems \cite{BergCKO08}
(e.g., windowing) use interval search as a module.
\rev{The
internal states of window queries in Stream processors
(e.g. Flink/Kafka) can be modeled and
managed as intervals \cite{Awad0KVS20}.}

We study the classic problem of indexing a large collection $\mathcal{S}$ of objects (or records), based on an interval attribute that characterizes each object. Hence, we model each object $s\in \mathcal{S}$ as a triple $\langle s.id, s.st, s.end\rangle$, where $s.id$ is the object's identifier (which can be used to access any other attribute of the object), and $[s.st,s.end]$ is the interval associated to $s$.
Our focus is on interval {\em range} queries, the most fundamental query type over intervals. Given a query interval $q=[q.st,q.end]$, the objective is to find the ids of all objects $s \in \mathcal{S}$, whose intervals {\em overlap with} $q$. Formally, the result of a range query $q$ on object collection $\mathcal{S}$ is $\{s.id~|~s\in \mathcal{S} \wedge (s.st\le q.st\le s.end \lor q.st\le s.st \le q.end)\}$.
Range queries are also known as {\em time travel} or {\em timeslice} queries in temporal databases \cite{SalzbergT99}. Examples
of such queries
on different data domains include the following:
\begin{itemize}
\item on a relation storing employment periods: {\em find the employees who were employed sometime in $[1/1/2021,2/28/2021]$}.
\item on weblog data: {\em find the users who were active sometime between 10:00am and 11:00am yesterday}.
\item on taxi trips data: {\em find the taxis which were active (on a trip) between 15:00 and 17:00 on 3/3/2021}.
\item on uncertain temperatures: {\em find all stations having temperature between 6 and 8 degrees with a non-zero probability}.
\end{itemize}
Range queries
can be specialized to retrieve intervals that satisfy any relation in Allen's set \cite{Allen81}, e.g., intervals that are {\em covered by} $q$.
{\em Stabbing} queries
({\em pure-timeslice} queries in temporal databases)
are a special class of range queries for which $q.st=q.end$.
Without loss of generality, we assume that the intervals and queries are {\em closed} at both ends.
Our method can easily be adapted to manage intervals and/or process range queries, which are open at either or both sides, i.e., $[o.st,o.end)$, $(o.st,o.end]$ or $(o.st,o.end)$.

For efficient range and stabbing queries over collections of intervals, classic data structures for managing intervals, like the interval tree \cite{Edels80}, are typically used.
Competitive indexing methods include the timeline index \cite{KaufmannMVFKFM13}, 1D-grids and the period index \cite{BehrendDGSVRK19}.
All these methods, which we review in detail in Section \ref{sec:related}, have not been optimized for handling very large collections of intervals in main memory.
Hence, there is room for 
new data structures, which exploit the characteristics and capabilities of
modern machines that have large enough memory capacities
for the scale of data found in most applications.

\stitle{Contributions.}
In this paper, we propose a novel and general-purpose
Hierarchical index for INTervals (HINT), 
suitable for 
applications that manage large collections of intervals.
HINT defines a hierarchical decomposition of the domain and assigns each interval in $\mathcal{S}$ to at most two partitions per level.
If the domain is relatively small and discrete, our index can process interval range queries with no comparisons at all.
For the general case where the domain is large and/or continuous, we propose a generalized version of HINT,
denoted by HINT$^m$, which limits the number of levels to $m+1$ and greatly reduces the space requirements.
HINT$^m$ conducts comparisons only for the intervals in the first and last accessed partitions at the bottom levels of the index.  
Some of the unique and novel characteristics of our index include:
\begin{itemize}
\item The intervals in each partition are further divided into groups, based on whether they begin inside or before the partition. This division (1) cancels the need for detecting and eliminating duplicate query results, (2) reduces the data accesses to the absolutely necessary, and (3) minimizes the space needed for storing the objects into the partitions.
\item As we theoretically prove, the expected number of HINT$^m$ partitions for which comparisons are necessary is just four. This guarantees fast retrieval times, independently of the query extent\eat{length} and position.
\item The optimized version of our index stores the intervals in all partitions at each level sequentially and uses a dedicated array with just the ids of intervals there, as well as 
links between non-empty partitions at each level. These optimizations facilitate sequential access to the query results at each level, while avoiding accessing unnecessary data.
\end{itemize}

Table \ref{tab:indexes} qualitatively compares HINT  to
previous work.
Our experiments\eat{al evaluation} on real and synthetic datasets show\eat{s} that our index is {\em one order of magnitude faster} than the competition.
As we explain in Section \ref{sec:related}, existing indices typically require at least one comparison for each query result (interval tree, 1D-grid) or may access and compare more data than necessary (timeline index, 1D-grid).
\eat{In addition}Further, the 1D-grid, the timeline and the period index need more space than HINT in the presence of long intervals in the data due to\eat{they do} excessive replication either in their partitions (1D-grid, period index) or their checkpoints (timeline index).
HINT  gracefully supports updates, since each partition (or division within a partition) is independent from others.
The construction cost of HINT is also low, as we verify experimentally.
Summing up, HINT is superior in all aspects \eat{compared }to the state-of-the-art and constitutes an important contribution,
given the fact that range queries over large collections of intervals is a fundamental problem with numerous applications.

\eat{
\begin{table}
	\centering
	\caption{Comparison of interval indices}
	\label{tab:indexes}
	\vspace{-1ex}
	\small
	\begin{tabular}{|l|@{~}c@{~}|@{~}c@{~}|@{~}c@{~}|}\hline
		{\bf Method}&comparisons									&space	&updates										\\\hline\hline
		
		Interval tree \cite{Edels80}&$O(logn+K)$
		&$O(n)$&no/slow\\
		Timeline index \cite{KaufmannMVFKFM13}	
		&$O(\Delta)$
		&$O(n+chk\cdot C)$&no/slow\\ 	
		1D-grid 
		&$O(partsize)$
		&$O(n\cdot rep)$&yes/fast\\ 		
Period index \cite{BehrendDGSVRK19}	
&$O(partsize)$
&$O(n\cdot rep)$&yes/fast\\ 	
HINT/HINT$^m$ (our work)	
&$O(Hpartsize)$
&$O(nlogm)$&yes/fast\\ 	
		\hline
	\end{tabular}
      \end{table}
}
    
\begin{table}
	\centering
	\caption{Comparison of interval indices}
	\label{tab:indexes}
	\vspace{-1ex}
	\footnotesize
	\begin{tabular}{|l|@{~}c@{~}|@{~}c@{~}|@{~}c@{~}|}\hline
		{\bf Method}&{\bf query cost} & {\bf space}	&{\bf updates}										\\\hline\hline
		
		Interval tree \cite{Edels80}&medium
		&low&slow\\
		Timeline index \cite{KaufmannMVFKFM13}	
		&medium
		&medium&slow\\ 	
		1D-grid
		&medium
		&medium&fast\\ 		
Period index \cite{BehrendDGSVRK19}	
&medium
&medium&fast\\ 	
HINT/HINT$^m$ (our work)
&low
&low&fast\\ 	
		\hline
	\end{tabular}
	\vspace{-3ex}
      \end{table}

\stitle{Outline.}
Section \ref{sec:related} reviews related work and presents in detail the characteristics and  weaknesses of existing interval indices.
In Section \ref{sec:hierarchical}, we present HINT and its generalized HINT$^m$ version and analyze their complexity. 
Optimizations that boost the performance of  HINT$^m$ are presented in Section \ref{sec:opts}.
We evaluate the performance of HINT$^m$  experimentally in Section \ref{sec:exps} on real and synthetic data and compare it to the state-of-the-art. Finally, Section \ref{sec:concl} concludes the paper with a discussion about future work.

\section{Related Work}
\label{sec:related}
In this section, we present in detail the state-of-the-art main-memory indices for intervals, to which we experimentally compare HINT in Section \ref{sec:exps}.
In addition, we briefly discuss other relevant data structures and previous work on other queries over interval data.   


\stitle{Interval tree.}
One of the most popular data structures for intervals 
is Edelsbrunner's {\em interval tree} \cite{Edels80},
a binary search tree, which takes $O(n)$ space and answers queries in $O(\log n+K)$ time ($K$ is the number of query results). The tree divides the domain hierarchically by placing all intervals strictly before (after) the domain's center to the left (right) subtree and all intervals that overlap with the   domain's center at the root. This process is repeated recursively for the left and right subtrees using the centers of the corresponding sub-domains. The intervals assigned to each tree node are sorted in two lists based on their starting and ending values, respectively.
Interval trees are used to answer {\em stabbing} and interval (i.e., {\em range}) queries.
For example, Figure~\ref{fig:itree} shows a set of 14 intervals $s_1,\ldots,s_{14}$, which are assigned to 7 interval tree nodes and a query interval $q=[q.st,q,end]$. The domain point $c$ corresponding to the tree's root is {\em contained in} the query interval, hence all intervals in the root are reported and both the left and right children of the root have to be visited recursively. 
Since the left child's point $c_L$ is before $q.st$, we access the END list from the end and report results until we find an interval $s$ for which $s.end<q.st$; then we access recursively the right child of $c_L$. This process is repeated symmetrically for the root's right child $c_R$.
The main drawback of the interval tree is that we need to perform comparisons for most of the intervals in the query result.
In addition, updates on the tree can be slow because the lists at each node should be kept sorted.
A relational interval tree for disk-resident data was proposed in \cite{KriegelPS00}. 

\begin{figure}[t]
	\includegraphics[width=0.9\columnwidth]{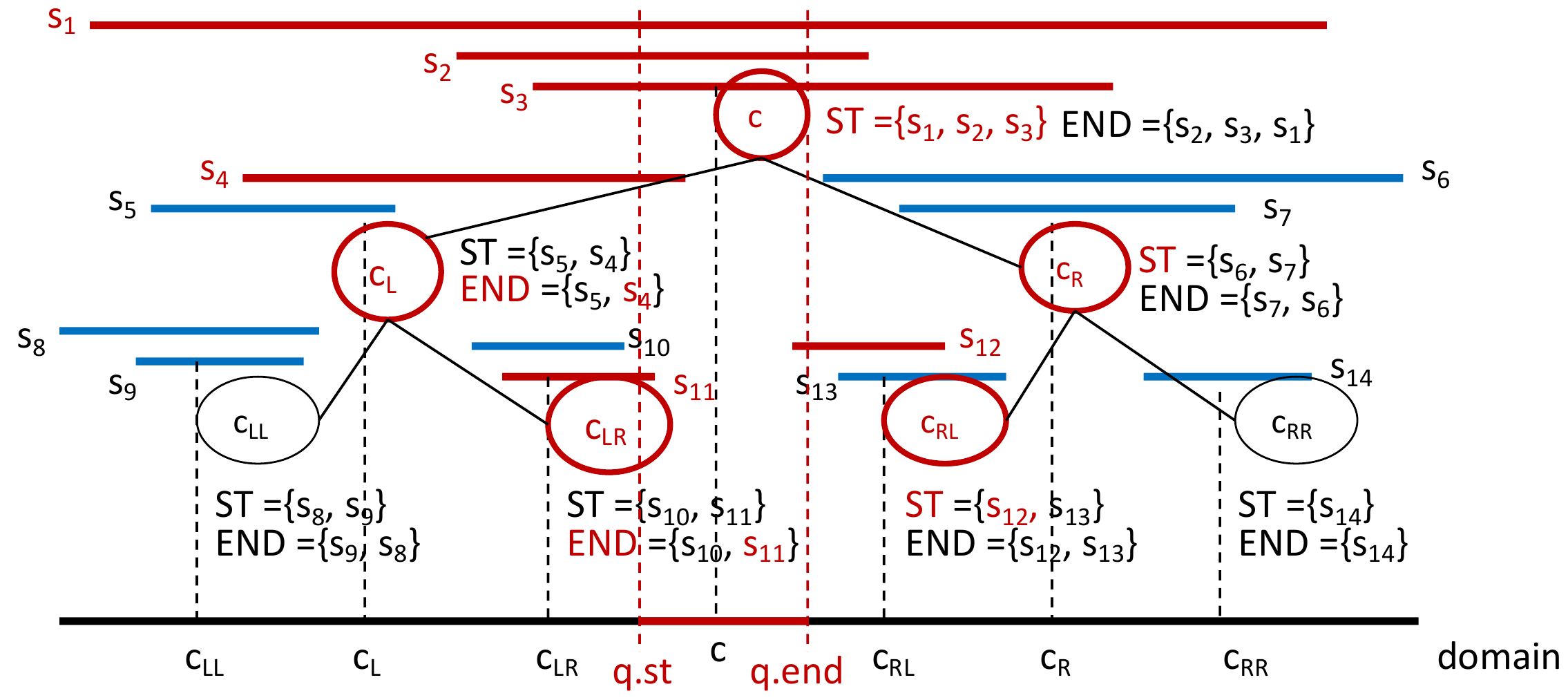}
	\vspace{-2mm}
	\caption{Example of an interval tree}
	\vspace{-3ex}
	\label{fig:itree}
\end{figure}

\stitle{Timeline index.}
The timeline index \cite{KaufmannMVFKFM13} is
a general-purpose access method
for temporal (versioned) data, implemented in SAP-HANA.
It keeps the endpoints of all intervals in an {\em event list}, which is a table of $\langle time, id, isStart \rangle$ triples, where $time$ is the value of the start or end point of the interval, $id$ is the identifier of the interval, and $isStart$
1 or 0,  depending on whether $time$ corresponds to the start or end of the interval, respectively. 
The event list is sorted primarily by $time$ and secondarily by $isStart$ (descending).
In addition, at certain timestamps, called {\em checkpoints}, the entire set of {\em active} object-ids
is materialized,
that is the intervals that contain the checkpoint.
For each checkpoint, there is a link
to the first triple in the event list for which
$isStart$=0 and $time$ 
is greater than or equal to the checkpoint,
Figure~\ref{fig:timeline}(a) shows a set of five intervals ($s_1,\ldots,s_5$) and Figure~\ref{fig:timeline}(b) exemplifies a timeline index for them.

\eat{In order t}To evaluate a range query (called {\em time-travel} query in \cite{KaufmannMVFKFM13}), we \eat{should }first find the largest checkpoint which is smaller than or equal to $q.st$
(e.g., $c_2$ in Figure~\ref{fig:timeline}) 
and initialize\eat{ an interval set} $R$ as the
active interval set at the checkpoint (e.g., $R=\{s_1,s_3,s_5\}$).
Then, we scan the event list from the position pointed by the checkpoint, 
until the first triple for which $time\ge q.st$,
and update $R$ by
inserting to it intervals corresponding to an $isStart=1$ event and deleting the ones 
corresponding to a $isStart=0$ triple
(e.g., $R$ becomes $\{s_3,s_5\}$).
When we reach $q.st$, all intervals in $R$ are guaranteed\eat{ to be part of the} query results and they are reported. We continue scanning the event list until the first triple after $q.end$ and we add to the result the ids of all intervals corresponding to triples with $isStart=1$ (e.g., $s_2$ and $s_4$).


\begin{figure}[t]
\begin{tabular}{cc}
\hspace{-1ex}\includegraphics[width=0.4\columnwidth]{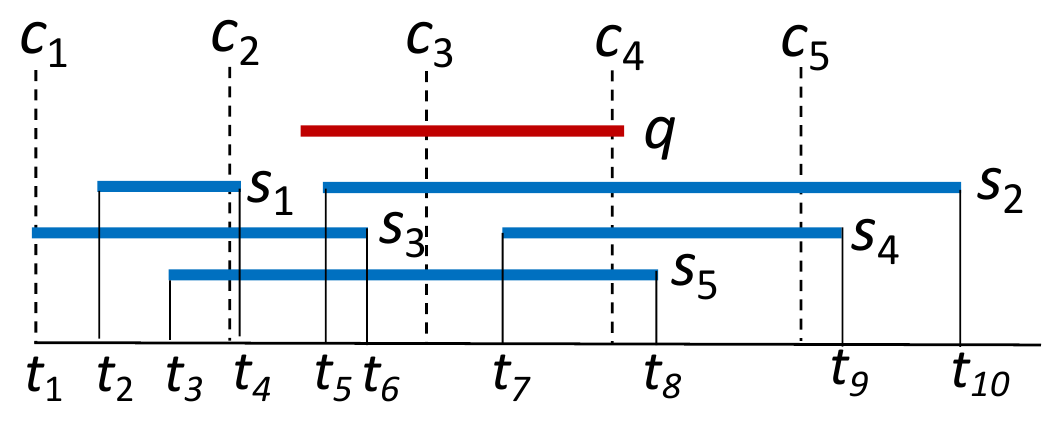}
&\includegraphics[width=0.55\columnwidth]{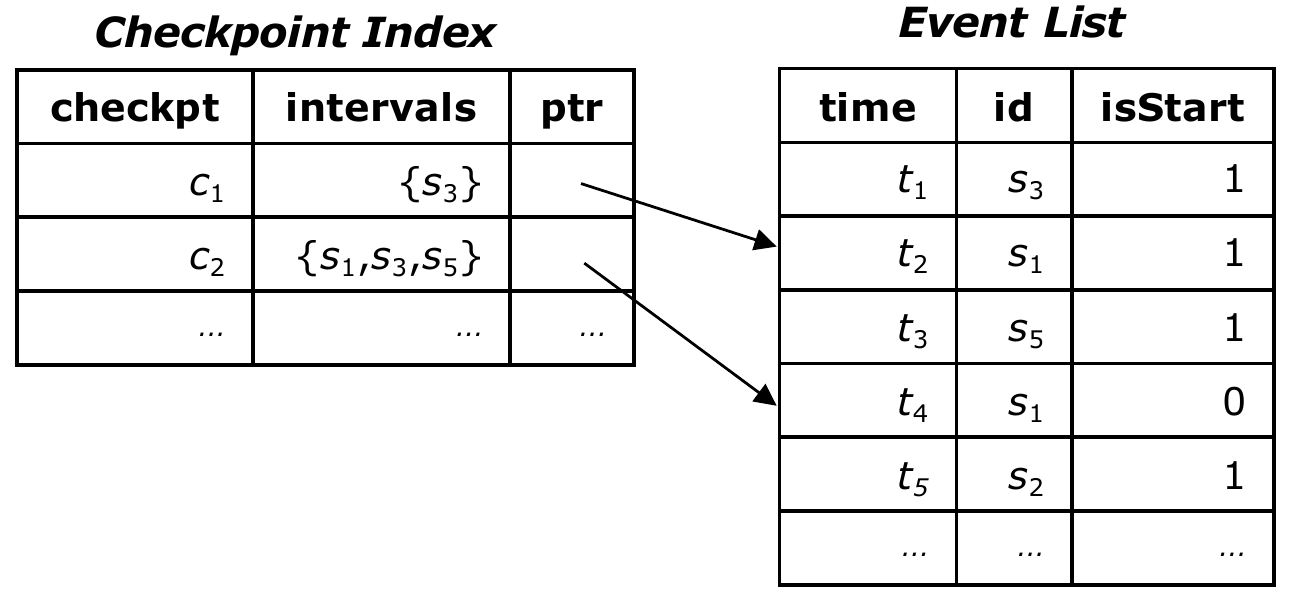}
\\
(a) set of intervals & timeline index
\end{tabular}
\vspace{-2mm}
\caption{Example of a timeline index}
\vspace*{-3ex}
\label{fig:timeline}
\end{figure}

The timeline index accesses more data and performs more comparisons than necessary,
during range query evaluation.
The index also requires a lot of extra space to store the active sets of the checkpoints. Finally, ad-hoc updates are expensive because the event list should be kept sorted.


\stitle{1D-grid.}
A simple and practical data structure for intervals is a 1D-grid, which divides the domain into $p$ partitions
$P_1,P_2,\dots,P_p$.
The partitions
are pairwise disjoint in terms of their interval span and collectively cover the entire data domain $D$.
Each interval is assigned to all partitions that it overlaps with.
Figure~\ref{fig:flat} shows 5 intervals assigned to $p=4$ partitions;
$s_1$ goes to $P_1$ only, whereas $s_5$ goes to all four partitions.
Given a range query $q$, the results can be obtained by accessing each partition $P_i$ that overlaps with $q$. For each $P_i$ which is {\em contained in} $q$ (i.e., $q.st\le P_i.st \wedge P_i.end\le q.end$), all intervals in $P_i$ are guaranteed to overlap with $q$. For each $P_i$, which overlaps with $q$, but is not contained in $q$, we should compare each $s_i \in P_i$ with $q$ to determine whether $s_i$ is a query result.
If the interval of a range query $q$ overlaps with multiple partitions,
duplicate results may be produced\eat{, which need to be handled}. 
An efficient approach for handling duplicates 
is the {\em reference value} method \cite{DittrichS00}, which was originally proposed for rectangles but can be directly applied for 1D intervals. For each interval $s$ found to overlap with $q$ in a partition $P_i$, we compute $v=\max\{s.st, q.st\}$ as the {\em reference value} and report $s$ only if $v\in [P_i.st,P_i.end]$. Since $v$ is unique, $s$ is reported only in one partition.
In Figure~\ref{fig:flat}, interval $s_4$ is reported only in $P_2$ which\eat{as $P_2$} contains value $\max\{s_4.st, q.st\}$.

The 1D-grid has two drawbacks. First, the duplicate results should be computed and checked before being eliminated by the reference value\eat{ approach}. Second, if the \eat{indexed }collection contains many long intervals,
the index may grow large in size due to excessive replication which increases the number of duplicate results to be eliminated. 
In contrast\eat{On the other hand}, \eat{the }1D-grid supports fast updates as the partitions are stored independently with\eat{and there is} no need to organize the intervals in them.

\begin{figure}[t]
     \includegraphics[width=0.85\columnwidth]{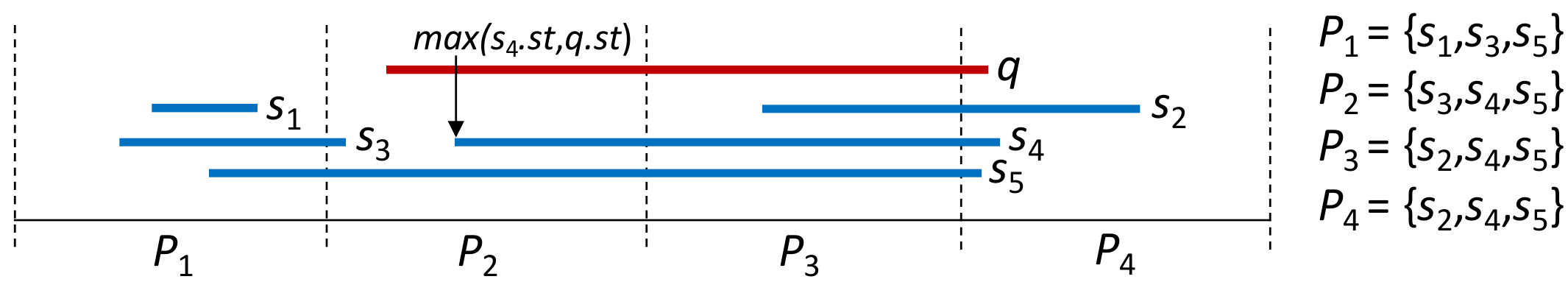}
     \vspace{-2mm}
     \caption{Example of a 1D-grid}
  \label{fig:flat}
  \vspace{-2ex}
\end{figure}

\stitle{Period index.}
The {\em period index} \cite{BehrendDGSVRK19} is a domain-partitioning self-adaptive structure, specialized for {\em range} and {\em duration} queries.
The time domain is split into coarse partitions as in a 1D-grid and then each partition is divided hierarchically, in order to organize the intervals assigned to the partition based on their positions and durations. Figure~\ref{fig:pindex} shows a set of intervals and how they are partitioned in a period index. There are two primary partitions $P_1$ and $P_2$ and each of them is divided hierarchically to three levels. Each level corresponds to a duration length and each interval is assigned to the level corresponding to its duration. The top level stores intervals shorter than the length of a division there, the second level stores longer intervals but shorter than a division there, and so on.
Hence, each interval is assigned to at most two divisions, except for intervals which are assigned to the bottom-most level, which can go to an arbitrary number of divisions. 
During query evaluation, only the divisions that overlap the query range are accessed; if the query carries a duration predicate, the divisions that are shorter than the query duration are skipped.   
For range queries, the period index performs in par with the interval tree and the 1D-grid \cite{BehrendDGSVRK19}, so we also compare against this index in Section \ref{sec:exps}.  

\begin{figure}[t]
     \includegraphics[width=0.75\columnwidth]{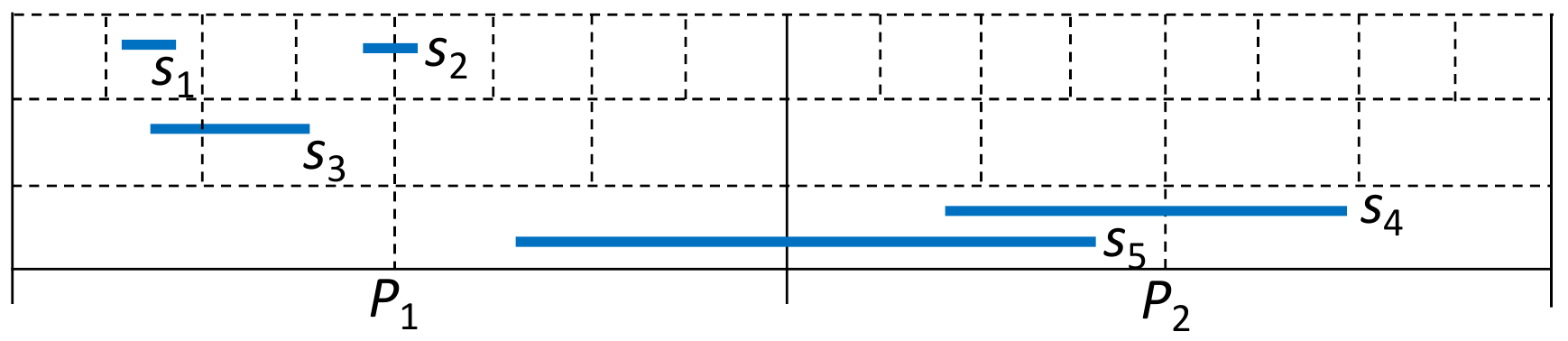}
     \vspace{-2mm}
     \caption{Example of a period index}
     \vspace{-3ex}
  \label{fig:pindex}
\end{figure}


\stitle{Other \eat{related }works.}
Another classic data structure for intervals is the {\em segment tree} \cite{BergCKO08},
a binary search tree, which
has $O(n\log n)$ space complexity and answers stabbing queries in $O(\log n+K)$ time.
The segment tree is not designed for range queries, for which it requires a duplicate result elimination mechanism.
In computational geometry \cite{BergCKO08}, indexing intervals has been studied as a subproblem within 
orthogonal 2D range search, and the worst-case optimal interval tree is typically used. Indexing intervals has re-gained interest with the advent of temporal databases \cite{BohlenDGJ17}.
\eat{In the context of}For temporal data, a number of indices are proposed for secondary memory, mainly for\eat{focus on} effective versioning and compression \cite{BeckerGOSW96,LometHNZ08}. 
Such 
indexes are tailored for historical versioned data,
while we focus on arbitrary interval sets, queries, and updates.

Additional research on indexing intervals does not address 
range queries, but other operations such as {\em temporal aggregation} \cite{KlineS95,MoonLI03,KaufmannMVFKFM13} and {\em interval joins}
\cite{DignosBG14,PiatovHD16,BourosM17,BourosLTMT20,BourosMTT21,PiatovHDP21,CafagnaB17,ChekolPS19,ZhuFYPW19}. The timeline index \cite{KaufmannMVFKFM13} can  be directly used for\eat{to perform} temporal aggregation.
Piatov et al. \cite{PiatovH17} present a collection of plane-sweep algorithms that extend the timeline index 
to support aggregation over fixed intervals, sliding window aggregates, and MIN/MAX aggregates. 
The timeline index was later adapted
for interval overlap joins \cite{PiatovHD16,PiatovHDP21}.
A {\em domain partitioning} technique for parallel processing of interval joins was proposed in \cite{BourosM17,BourosLTMT20,BourosMTT21}.
Alternative partitioning techniques for \eat{overlap }interval joins were proposed in \cite{DignosBG14,CafagnaB17}.
Partitioning techniques for interval joins cannot replace interval indices as they are not designed for
range queries.
Temporal joins considering Allen's algebra relationships for RDF data were studied in \cite{ChekolPS19}.
Multi-way interval joins in the context of temporal $k$-clique enumeration were studied in \cite{ZhuFYPW19}.
\rev{Awad et al. \cite{Awad0KVS20} define {\em interval events} in data streams by events of the same or different types that are observed in succession. Analytical operations based on aggregation or reasoning operations can be used to formulate composite interval events.}


\section{HINT}
\label{sec:hierarchical}
In this section, we propose the
\emph{Hierarchical index for INTervals} or  HINT,
which defines a hierarchical domain decomposition and assigns each interval to at most two partitions per level.
The primary goal of the index is to minimize the number of comparisons during query evaluation, while keeping the space requirements relatively low, even when there are long intervals in the collection.
HINT applies a smart division of intervals in each
partition into two groups, which avoids the production and handling of
duplicate query results and minimizes the number of intervals that have to
be accessed.
In Section~\ref{sec:hierarchical:precise}, we present a version of
HINT, which
avoids comparisons overall during query evaluation, but it is not
always applicable and may have high space requirements.
Then, Section~\ref{sec:hierarchical:partial}  presents HINT$^m$,
the general version of our index,
used for intervals in arbitrary domains.
\rev{
Last, Section~\ref{sec:hierarchical:m} describes our analytical model for setting the $m$ parameter and Section~\ref{sec:hierarchical:updates} discusses updates. 
}
Table~\ref{table:notations} summarizes the notation used
in the paper.




\begin{table}[t]
\caption{Table of notations\label{table:notations}}
\vspace{-2ex}
\centering
\footnotesize
\begin{tabular}{|c|c|}
\hline
\textbf{notation} &\textbf{description}\\ %
\hline\hline
$s.id, s.st, s.end$ &identifier, start, end point of interval $s$ \\
$q=[q.st, q.end]$ &query range interval\\
$prefix(k, x)$ & $k$-bit prefix of integer $x$\\  
$P_{\ell,i}$ & $i$-th partition at level $\ell$ of HINT/HINT$^m$\\
$P_{\ell,f}$ ($P_{\ell,l}$) & first (last) partition at level $\ell$
                              that overlaps with $q$\\
$P_{\ell,i}^O$ ($P_{\ell,i}^R$) & sub-partition of $P_{\ell,i}$ with originals (replicas)\\
\rule{0pt}{4ex}  $P_{\ell,i}^{O_{in}}$ ($P_{\ell,i}^{O_{aft}}$) &
                                                                  intervals in $P^O_{\ell,i}i$ ending inside (after) the partition\\
\hline
\end{tabular}
\vspace*{-3ex}
\end{table}

\subsection{A comparison-free version of HINT}
\label{sec:hierarchical:precise}
 
We first describe a version of HINT,
which is appropriate in the case of a {\em discrete}
and {\em not very large} domain $D$.
Specifically, assume that the domain $D$ wherefrom the 
endpoints of intervals in $\mathcal{S}$ take value is $[0,2^m\!-\!1]$.
We can define a regular hierarchical decomposition of the domain into
partitions, where at each level $\ell$ from $0$ to $m$, there are $2^{\ell}$
partitions, denoted by array $P_{\ell,0},\eat{P_{\ell,1},}\dots,P_{\ell,2^{\ell}\!-\!1}$.
Figure~\ref{fig:exact1} illustrates the hierarchical domain
partitioning for $m=4$. 

Each interval $s\in S$ is assigned
to the {\em smallest set of partitions} which collectively define
$s$. 
It is not hard to show that $s$ will be assigned to at most two
partitions per level. For example, in Figure~\ref{fig:exact1},
interval $[5,9]$ is assigned to one
partition at level $\ell=4$ and two partitions at level $\ell=3$.
The
assignment procedure is described by 
Algorithm~\ref{algo:assign}. In a nutshell, for an interval $[a,b]$, starting from the
bottom-most level $\ell$, if the last bit of $a$ (resp. $b$) is 1 (resp. 0),
we assign the interval to partition $P_{\ell,a}$ (resp. $P_{\ell,b}$) and increase $a$
(resp. decrease $b$) by one.
We then update $a$ and $b$ by cutting-off their
last bits (i.e., integer division by 2, or bitwise right-shift). If, at the next level,
$a>b$ holds, indexing $[a,b]$ is done. 

\begin{algorithm}[t]
\LinesNumbered
\footnotesize
\Input{HINT index \HIDX, interval $s$}
\Output{updated \HIDX after indexing $s$}
$a\leftarrow s.st$; $b\leftarrow s.end$\comm*{set masks to $s$ endpoints}
$\ell \leftarrow m$\comm*{start at the bottom-most level}
\While {$\ell\ge 0$ \textbf{\emph{and}} $a\le b$}
{
  \If{last bit of $a$ is $1$}
  {
    \textbf{add} $s$ to $\HIDX.P_{\ell,a}$\comm*{update partition}
    $a \leftarrow a + 1$\;
  }
  \If{last bit of $b$ is $0$}
  {
    \textbf{add} $s$ to $\HIDX.P_{\ell,b}$\comm*{update partition}
    $b \leftarrow b-1$\;
  }
  $a \leftarrow a \div 2$; $b \leftarrow b \div 2$\comm*{cut-off last bit}
  $\ell\leftarrow\ell -1$\comm*{repeat for previous level}
}
\caption{Assignment of an interval to partitions}
\label{algo:assign}
\end{algorithm}

\begin{figure}[htb]
     \includegraphics[width=0.8\columnwidth]{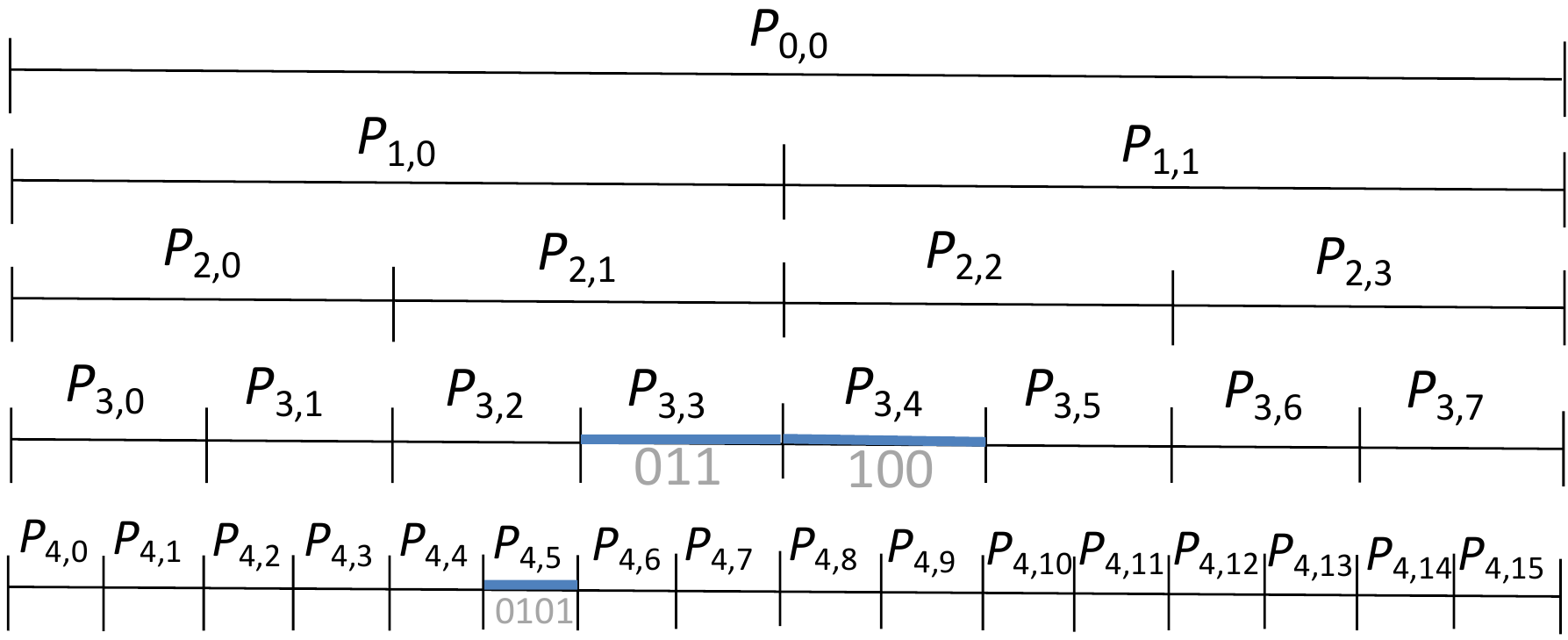}
\vspace{-2mm}
     \caption{Hierarchical partitioning and assignment of $[5,9]$}
   \label{fig:exact1}
   \vspace{-2ex}
\end{figure}

\subsubsection{Range queries}
A  range query $q$ can be evaluated by finding at each level the
partitions that overlap with $q$.
Specifically, the partitions that overlap with 
the query interval $q$ at level $\ell$ are partitions $P_{\ell,prefix(\ell,q.st)}$ to
$P_{\ell,prefix(\ell,q.end)}$, where $prefix(k,x)$ denotes the $k$-bit
prefix of integer $x$.
\rev{We call these partitions {\em relevant} to the query $q$.}
All intervals in \rev{the relevant} partitions are
guaranteed to overlap with $q$ and
intervals in none of these partitions cannot possibly overlap with $q$. 
However, since the same interval $s$ may
exist in multiple partitions that overlap with a query,
$s$ may be reported multiple times in the query result.

We propose a technique that avoids the production and therefore, the need for elimination of duplicates and, at the same time, minimizes the number of data accesses.
For this, we divide the intervals in each partition $P_{\ell,i}$ into
two groups: {\em originals} $P^O_{\ell,i}$ and
{\em replicas} $P^R_{\ell,i}$.
Group $P^O_{\ell,i}$ contains all intervals $s\in P_{\ell,i}$ that {\em begin} at
$P_{\ell,i}$
i.e.,  $prefix(\ell,s.st)=i$.
Group $P^R_{\ell,i}$ contains all intervals $s\in P_{\ell,i}$ that begin
before $P_{\ell,i}$,
i.e., $prefix(\ell,s.st)\ne i$.%
\footnote{Whether an interval $s\in P_{\ell,i}$ is assigned to $P^O_{\ell,i}$ or
$P^R_{\ell,i}$ is determined at insertion time (Algorithm~\ref{algo:assign}). At the first time Line~5 is executed, $s$ is added
as an original and in all other cases as a replica. If Line~5 is never executed, then $s$ is added as original the only time that Line~8 is executed.}
Each interval is added as original in only one partition of HINT.
For example, interval $[5,9]$ in Figure~\ref{fig:exact1} is added to 
$P^O_{4,5}$, $P^R_{3,3}$, and $P^R_{3,4}$.

Given a range query $q$, at
each level $\ell$ of the index, we report all intervals 
in the first
\rev{relevant} partition $P_{\ell,f}$ 
(i.e., $P^O_{\ell,f} \cup P^R_{\ell,f}$). Then, for every other \rev{relevant} 
partition $P_{\ell,i}$, $i>f$,
we report all
intervals in $P^O_{\ell,i}$ and ignore $P^R_{\ell,i}$. This guarantees
that no result is missed and no duplicates are produced. The reason is
that each interval $s$ will appear as original in just one partition,
hence, reporting only originals cannot produce any duplicates. At the
same time, all replicas $P^R_{\ell,f}$  in the first partitions per level $\ell$ that
overlap with  $q$ begin {\em before} $q$ and overlap with  $q$, so they should be
reported. On the other hand, replicas $P^R_{\ell,i}$ in subsequent
\rev{relevant} partitions ($i>f$) 
contain intervals, which are
either originals in a previous partition $P_{\ell,j}$, $j<i$ or
replicas in $P^R_{\ell,f}$, so, they can safely be skipped.
Algorithm~\ref{algo:queryexact}
describes the range query algorithm using HINT.

\begin{algorithm}[t]
\LinesNumbered
\footnotesize
\Input{HINT index \HIDX, query interval $q$}
\Output{set $\mathcal{R}$ of all intervals that overlap with  $q$}
$\mathcal{R}\leftarrow \emptyset$\;
\ForEach{level $\ell$ in \HIDX}
{
  $p \leftarrow prefix(\ell,q.st)$\;
  $\mathcal{R}\leftarrow \mathcal{R} \cup \{s.id | s\in \HIDX.P^O_{\ell,p} \cup \HIDX.P^R_{\ell,p}\}$\\
  \While{$p<prefix(\ell,q.end)$}
  {
    \textbf{set} $p\leftarrow p+1$\;
    $\mathcal{R}\leftarrow \mathcal{R} \cup \{s.id | s\in \HIDX.P^O_{\ell,p} \}$\\
  }
}
\Return $\mathcal{R}$\;
\caption{Range query \eat{algorithm }on HINT}
\label{algo:queryexact}
\end{algorithm}

For example, consider the hierarchical partitioning of Figure~\ref{fig:exactrange} and a query interval $q=[5,9]$. The binary
representations of $q.st$ and $q.end$ are 0101 and 1001, respectively.
The \rev{relevant} partitions 
at each level are shown in
bold (blue) and dashed (red) lines and can be determined by the
corresponding prefixes of  0101 and 1001. At each level $\ell$, {\em all}
intervals (both originals and replicas) in the first partitions $P_{\ell,f}$ (bold/blue) are
reported while in the subsequent partitions (dashed/red), {\em only} the {\em original} intervals are.

\begin{figure}[htb]
     \includegraphics[width=0.8\columnwidth]{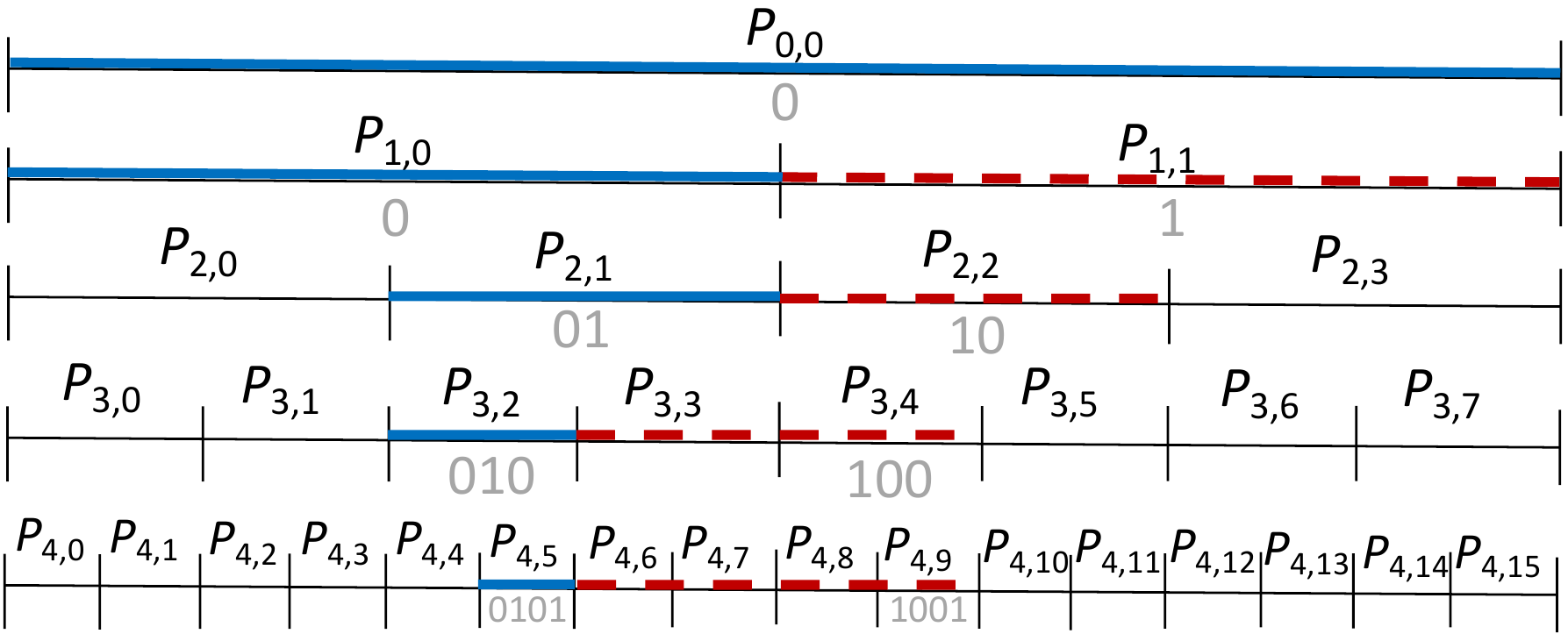}
     \vspace{-2mm}
     \caption{Accessed partitions for range query $[5,9]$}
  \label{fig:exactrange}
  \vspace{-2ex}
\end{figure}

\stitle{Discussion.}
The version of HINT described above finds all
range query results, without conducting any comparisons. 
This means
that in each partition $P_{\ell,i}$, we only have to keep the ids of
the intervals that are assigned to $P_{\ell,i}$ and do not have
to store/replicate the interval endpoints. 
In addition,
the relevant partitions at each level are computed by fast bit-shifting
operations which are comparison-free.
To use HINT for arbitrary integer domains, we should first
normalize all interval endpoints by subtracting the minimum endpoint,
in order to convert them to values in a $[0,2^m-1]$ domain (the same
transformation should be applied on the queries).
If the required $m$ is very large, we can index the intervals based on
their $m$-bit prefixes and support approximate search on discretized data.
Approximate
search can also be applied on intervals in a real-valued domain, after
rescaling and discretization in a similar way.


\subsection{HINT$^m$: indexing arbitrary intervals}\label{sec:hierarchical:partial}
We now present a generalized version of HINT, denoted by HINT$^m$,
which
can be used for intervals in arbitrary
domains.
HINT$^m$ uses a hierarchical domain partitioning with $m+1$ levels,
based on a $[0,2^m-1]$ domain $D$;
each raw interval endpoint is {\em mapped} to a
value in $D$, by linear rescaling. The mapping function
$f(\mathbb{R}\to D)$ is $f(x) = \lfloor
\frac{x-min(x)}{max(x)-min(x)} \cdot  (2^m-1)\rfloor$, where $min(x)$
and $max(x)$ are the minimum and maximum interval endpoints in the
dataset $S$, respectively.
Each raw interval $[s.st,s.end]$ is mapped to interval
$[f(s.st),f(s.end)]$.
The mapped interval is then assigned to at most two partitions per
level in HINT$^m$, using
Algorithm~\ref{algo:assign}.

For the ease of presentation, we will assume that the raw interval
endpoints take values in $[0,2^{m'}-1]$, where $m'>m$, which means
that the mapping function $f$ simply outputs the $m$ most significant
bits of its input. As an example, assume that $m=4$ and
$m'=6$. Interval $[21,38]$
(=$[0b010101,0b100110]$)
is mapped to interval
$[5,9]$
(=$[0b0101,0b1001]$)
and assigned to partitions $P_{4,5}$,
$P_{3,3}$, and $P_{3,4}$, as shown in Figure \ref{fig:exact1}.
\eat{
Intervals $[20,37]$
(=$[0b010100,0b100101]$)
and $[22,39]$
(=$[0b010110,0b100111]$)
also go to exactly the same set of partitions.}
So, in contrast to HINT,
the set of partitions whereto an interval $s$ is assigned in HINT$^m$ does not
define $s$, but the smallest
interval in the $[0,2^{m}-1]$ domain $D$, which {\em covers} $s$.
As in HINT, at each level $\ell$, we
divide each partition 
$P_{\ell,i}$ to $P^O_{\ell,i}$ and $P^R_{\ell,i}$, 
to avoid duplicate query results.




\subsubsection{Query evaluation using HINT$^m$}
\label{sec:hierarchical:qeval}
For a range query $q$,
simply reporting all intervals in the relevant partitions at
each level (as in Algorithm~\ref{algo:queryexact}) would produce false
positives.
Instead,
comparisons to the query endpoints may be required for
the first and the last partition at each level that overlap with $q$.
Specifically, we can consider each level of HINT$^m$ as a
1D-grid (see Section~\ref{sec:related})
and go
through the partitions at each level $\ell$ that overlap with  $q$.
For the first partition
$P_{\ell,f}$, we verify whether $s$ overlaps with $q$ for each
interval $s\in P^O_{\ell,f}$ and each $s\in P^R_{\ell,f}$.
For the last partition $P_{\ell,l}$, we verify whether
$s$ overlaps with $q$ for each interval $s\in P^O_{\ell,l}$.
For each partition $P_{\ell,i}$ between $P_{\ell,f}$ and $P_{\ell,l}$, we
report all $s\in P^O_{\ell,i}$ without any comparisons.
As an example, consider the HINT$^m$ index and the
range query interval $q$ shown in Figure~\ref{fig:firstlast}.
The identifiers of the relevant partitions to $q$ are shown in the
figure (and also some indicative intervals that are assigned to these partitions).
At level $m=4$, we have to perform comparisons for all intervals in
the first relevant partitions $P_{4,5}$. In partitions
$P_{4,6}$,\ldots,$P_{4,8}$, we just report the originals in them as results,
while in partition $P_{4,9}$ we compare the start points of all
originals with $q$, before we can confirm whether they are results or not.
We can simplify the overlap tests at the first and the last
partition of each level $\ell$ based on the following:
\begin{lemma}\label{lem:compcut}
At every level $\ell$, each $s\in P^R_{\ell,f}$ is a query result
iff $q.st \le s.end$. If $l>f$, each $s\in P^O_{\ell,l}$ is a query result
    iff $s.st \le q.end$.
\end{lemma}
\rev{
\begin{proof}
For the first 
relevant partition $P_{\ell,f}$ at each 
level $\ell$, for each replica $s\in P^R_{\ell,f}$,
$s.st < q.st$,
so $q.st \le s.end$
suffices  as an overlap test.
For the last partition $P_{\ell,l}$,
if $l>f$, for each original $s\in P^O_{\ell,f}$, 
$q.st < s.st$, so
$s.st \le q.end$ suffices  as an overlap test.
\end{proof}
}

\begin{figure}[t]
     \includegraphics[width=0.8\columnwidth]{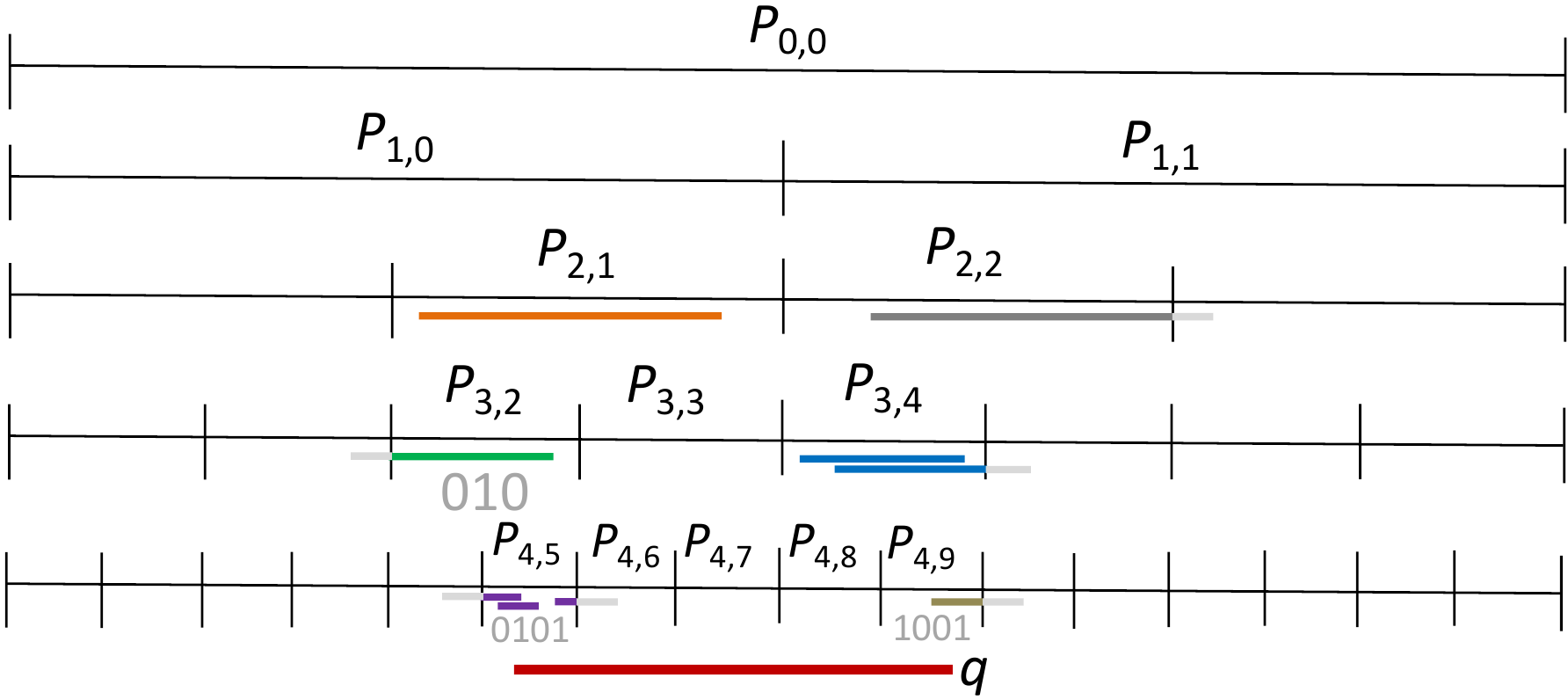}
     \vspace{-2mm}
     \caption{Avoiding redundant comparisons in HINT$^m$}
  \label{fig:firstlast}
  \vspace{-2ex}
\end{figure}

\subsubsection{Avoiding redundant comparisons in query evaluation}
One of our most important findings in this study and a powerful
feature of HINT$^m$ is that
at most levels, 
it is not necessary to do
comparisons at the first and/or the last partition.
For instance, in the previous example,
we do not have to perform comparisons for partition $P_{3,4}$, since
any interval assigned to $P_{3,4}$ should overlap with  $P_{4,8}$ and the
interval spanned by $P_{4,8}$ is covered by $q$. This means that the
start point of all intervals in $P_{3,4}$ is guaranteed to be before
$q.end$ (which is inside $P_{4,9}$). In addition, observe that for any
relevant partition which is the last partition at an upper level and
covers $P_{3,4}$ (i.e., partitions $\{P_{2,2},P_{1,1},P_{0,0}\}$),
we do not have to conduct
the $s.st\le q.end$ tests as intervals in these partitions are
guaranteed to start before  $P_{4,9}$.
The lemma below formalizes these observations:

\rev{
\begin{lemma}\label{lem:firstlast}
  If the first (resp. last) relevant partition for a query $q$ at level $\ell$ ($\ell<m$) starts (resp. ends) at
 the same value as the  first (resp. last) relevant partition at level
 $\ell+1$, then for every first (resp. last) relevant partition $P_{v,f}$ (resp. $P_{v,l}$) at levels $v\le
 \ell$, each interval $s\in P_{v,f}$ (resp. $s\in P_{v,l}$)  
 satisfies $s.end\ge q.st$ (resp. $s.st\le q.end$).
\end{lemma}
\begin{proof}
	Let $P.st$ (resp. $P.end$) denote the first (resp. last) domain value of partition $P$.
	Consider the first relevant partition $P_{\ell,f}$ at level $\ell$ and assume that 
	$P_{\ell,f}.st=P_{\ell+1,f}.st$. Then, for every interval $s\in P_{\ell,f}$, $s.end\ge P_{\ell+1,f}.end$, otherwise $s$ 
	would have been allocated to $P_{\ell+1,f}$ instead of $P_{\ell,f}$. Further, $P_{\ell+1,f}.end \ge q.st$, since 
	$P_{\ell+1,f}$ is the first partition at level $\ell+1$ which overlaps with $q$. Hence, $s.end\ge q.st$.
	Moreover, for every interval $s\in P_{v,f}$ with $v<\ell$, $s.end\ge P_{\ell+1,f}.end$ holds, as interval $P_{v,f}$ covers interval $P_{\ell,f}$; so, we also have $s.end\ge q.st$.
	Symmetrically, we prove that if $P_{\ell,l}.end=P_{\ell+1,l}.end$, then for each $s\in P_{v,l}, v\le \ell$, $s.st\le q.end$.
\end{proof}

}
We next focus on how to rapidly check the condition of 
Lemma~\ref{lem:firstlast}. Essentially, if the last bit of the
offset
$f$ (resp. $l$) of the first (resp. last) partition $P_{\ell,f}$
(resp. $P_{\ell,l}$) relevant to the query
at level $\ell$ is 0 (resp. 1), then \eat{this means that }the 
first (resp. last) partition at level $\ell-1$ above satisfies the
condition.
For example, in Figure~\ref{fig:firstlast},
consider the last relevant
partition $P_{4,9}$ at level $4$.
The last bit of $l=9$ is 1; so,
the last partition $P_{3,4}$ at level $3$ satisfies the condition and
we do not have to perform comparisons in the last 
partitions at level $3$ and above.

Algorithm~\ref{algo:rangeqhier} is a pseudocode for the range query algorithm on
HINT$^m$.
The algorithm accesses all levels of the index, bottom-up.
It uses two  auxiliary flag
variables $compfirst$ and $complast$ to mark whether it is
necessary to perform comparisons at the current level (and all levels
above it) at the first and the last partition,
respectively, according to the discussion in
the previous paragraph.
At
each level $\ell$, we find the 
offsets
of the relevant partitions
to the query, based on the $\ell$-prefixes of $q.st$ and $q.end$ (Line~\ref{lin:bounds}).
For the first position $f=prefix(q,st)$, the partitions holding
originals and replicas $P^O_{\ell,f}$ and $P^R_{\ell,f}$ are
accessed.
The algorithm
first checks whether $f=l$, i.e., the first and the last partitions
coincide.
In this case, if $compfirst$ and $complast$ are set, then we
perform all comparisons in $P^O_{\ell,f}$ and apply the first
observation in Lemma~\ref{lem:compcut} to $P^R_{\ell,f}$.
Else, if only $complast$ is set, we can safely skip the
$q.st\le s.end$ comparisons; if only $compfist$ is set, regardless
whether $f=l$, we just perform $q.st\le s.end$ comparisons to both
originals and replicas to the first partition.
Finally, if neither  $compfirst$ nor $complast$ are set, we just
report all intervals in the first partition as results.
If we are at the last partition $P_{\ell,l}$ and $l>f$ (Line 17) then  we just examine
$P^O_{\ell,l}$  and apply just the  $s.st\le q.end$ test for each
interval there, according to Lemma~\ref{lem:compcut}. Finally, for all partitions in-between the first
and the last one, we simply report all original intervals there. 

\begin{algorithm}[t]
\LinesNumbered
\footnotesize
\Input{HINT$^m$ index \HIDX, query interval $q$}
\Output{set $\mathcal{R}$ of intervals that overlap with  $q$}
$compfirst\leftarrow TRUE$; $complast \leftarrow TRUE$\;
$\mathcal{R}\leftarrow \emptyset$\; 
\For(\comm*[f]{bottom-up}){$\ell= m$ \textbf{to} $0$}
{
  $f \leftarrow prefix(\ell,q.st)$; $l \leftarrow prefix(\ell,q.end)$\; \label{lin:bounds}
  \For{$i=f$ to $l$}
  {
    \If(\comm*[f]{first overlapping partition}){$i=f$}
    {
      \If{$i=l$ {\bf and} $compfirst$ {\bf and} 
          $complast$}
        {
          $\mathcal{R}\leftarrow \mathcal{R} \cup \{s.id | s\in
          \HIDX.P^O_{\ell,i}, q.st \le s.end \land s.st\le q.end\}$\;
          $\mathcal{R}\leftarrow \mathcal{R} \cup \{s.id | s\in 
          \HIDX.P^R_{\ell,i}, q.st\le s.end\}$\;
        }
        \ElseIf{$i=l$ {\bf and} $complast$}
        {
         $\mathcal{R}\leftarrow \mathcal{R} \cup \{s.id | s\in \HIDX.P^O_{\ell,i}, s.st\le q.end\}$\;
          $\mathcal{R}\leftarrow \mathcal{R} \cup \{s.id | s\in \HIDX.P^R\}$\;
        }
        \ElseIf{$compfirst$}
        {
          $\mathcal{R}\leftarrow \mathcal{R} \cup \{s.id | s\in
          \HIDX.P^O_{\ell,i}\cup \HIDX.P^R_{\ell,i}, q.st\le s.end\}$\;
        \vspace*{-3ex}
        }
        \Else 
        {
          $\mathcal{R}\leftarrow	` \mathcal{R} \cup \{s.id | s\in \HIDX.P^O_{\ell,i} \cup
          \HIDX.P^R_{\ell,i}\}$\;
        }
     }
    \ElseIf(\comm*[f]{last partition, $l>f$}){$i=l$ {\bf and} $complast$}
    {
         $\mathcal{R}\leftarrow \mathcal{R} \cup \{s.id | s\in \HIDX.P^O_{\ell,i}, s.st\le q.end\}$\;
    }
    \Else(\comm*[f]{in-between or last ($l>f$), no comparisons})
     {
          $\mathcal{R}\leftarrow \mathcal{R} \cup \{s.id | s\in \HIDX.P^O_{\ell,i} \}$\;
     }
  }
  \If(\comm*[f]{last bit of $f$ is 0}){$f~\textrm{mod}~2 = 0$}
  {$compfirst \leftarrow FALSE$\;}
  \If(\comm*[f]{last bit of $l$ is 1}){$l~\textrm{mod}~2 = 1$}
  {$complast \leftarrow FALSE$\;}
}
\Return $\mathcal{R}$\;
\caption{Range query \eat{algorithm }on HINT$^m$}
\label{algo:rangeqhier}
\end{algorithm}

\rev{
\subsubsection{Complexity Analysis}
\label{sec:hint:analysis}
Let $n$ be the number of intervals in $\mathcal{S}$. 
Assume that the domain is $[0,2^{m'}-1]$, where $m'>m$.
To analyze the space complexity of HINT$^m$, we first prove the following lemma:

\begin{lemma}\label{lem:lastlevelsize}
	The total number of intervals assigned at the lowest level $m$ of
	HINT$^m$ is expected to be $n$.
\end{lemma}

\begin{proof}
        Each interval 
        $s\in \mathcal{S}$ will go to zero, one, or two partitions at level $m$, based\eat{depending} on the bits of $s.st$ and $s.end$ at position $m$ (see Algorithm~\ref{algo:assign}); 
	on average, $s$ will go to one partition.
\end{proof}


Using Algorithm \ref{algo:assign}, when an interval is assigned to a partition at a level $\ell$, the interval is {\em truncated} (i.e., shortened) by $2^{m'-\ell}$. Based on this, we analyze the space complexity of HINT$^m$ as follows.

\begin{theorem}\label{lem:spacecomp}
	Let $\lambda$ be the average length of intervals in input
        collection $S$.
        The space complexity of
	HINT$^m$ is $O(n\cdot\log_2 (2^{\log_2 \lambda -m'+m}+1))$.
\end{theorem}

\begin{proof}
	Based on Lemma \ref{lem:lastlevelsize}, each $s\in S$ will be assigned on average to one partition at level $m$ and will be truncated by $2^{m'-m}$. Following Algorithm \ref{algo:assign}, at the next level $m-1$, $s$ is also be expected to be assigned to one partition (see Lemma \ref{lem:lastlevelsize}) and truncated by $2^{m'-m+1}$, and so on, until the entire interval is truncated (condition $a\le b$ is violated at Line 3 of Algorithm \ref{algo:assign}). Hence, we are looking for the number of levels whereto each $s$ will be assigned, or for the
	smallest $k$ for which $2^{m'-m}+2^{m'-m+1}+\dots +2^{m'-m+k-1}\ge \lambda$.
      Solving the inequality gives $k\ge \log_2 (2^{\log_2 \lambda
        -m'+m}+1)$ and the space complexity of HINT$^m$ is $O(n\cdot k)$
\end{proof}
	
\eat{Regarding}For the computational cost of 
range queries in terms of conducted comparisons, 
in the worst case, $O(n)$ intervals are assigned to the first relevant partition $P_{m,f}$ at level $m$ and $O(n)$ comparisons are required. 
To estimate the {\em expected cost} of range query evaluation in
terms of conducted comparisons, we assume a uniform distribution of
intervals to partitions and random query intervals. 

\eat{
We first estimate
the number of intervals assigned to level $m$ of the index.

\begin{lemma}\label{lem:lastlevelsize}
	The total number of intervals assigned at the lowest level $m$ of
	HINT$^m$ is expected to be $n$.
\end{lemma}

\begin{proof}
	Let $y$ be the length of each partition at the lowest level $m$.
	Let $\lambda$ be the length of an interval $s\in \mathcal{S}$.
	If $\lambda\ge y$, 
	then $s$ is expected to fall in
	zero, one, or two partitions at level $m$ depending on the last indexable
	bits of $s.st$ and $s.end$; there are four combinations for the two
	bits and on average $s$ is going to be assigned to just one partition.
	If $\lambda< y$, $s$ either spans two
	partitions at level $m$ or just one. In the first case, the partitions are
	neighboring and $s$ will either fall in both or in neither of them
	(hence, one partition on average). In the second case, $s$ will fall
	into just the partition that includes it.
	Overall, each interval $s\in \mathcal{S}$ falls into exactly one partition at level $m$
	on average.
\end{proof}
}

\begin{lemma}\label{lem:expectedcompparts}
	The expected number of HINT$^m$ partitions for which
	we have to conduct comparisons is four.
\end{lemma}

\begin{proof}
At the last level of the index $m$, we definitely have to
do comparisons in the first and the last partition (which are
different in the worst case). At level $m-1$, for each of the first and last
partitions, we have a 50\% chance to avoid comparisons, due to Lemma
\ref{lem:firstlast}. Hence, the
expected number of partitions for which we have to perform comparisons
at level $m-1$ is 1. Similarly, at level $m-2$ each of the yet active first/last
partitions has a  50\% chance to avoid comparisons. Overall, for the
worst-case conditions, where $m$ is large and $q$ is long,
the expected number of partitions, for which we need to perform
comparisons is $2+1+0.5+0.25+\dots=4$.
\end{proof}	

\begin{theorem}\label{lem:expectedcom}
	The expected number of comparisons during query evaluation over HINT$^m$ is $O(n/2^m)$.
\end{theorem}

\begin{proof}
	For each query, we expect to conduct comparisons at least in the first and the last relevant partitions at level $m$.
	The expected number of intervals, in each
	of these two partitions, is $O(n/2^m)$, considering Lemma \ref{lem:lastlevelsize} and assuming a uniform distribution of the intervals in the partitions. In addition, due to Lemma \ref{lem:expectedcompparts}, the number of expected additional partitions that require comparisons is 2 and each of these two partitions is expected to also hold at most $O(n/2^m)$ intervals, by Lemma \ref{lem:lastlevelsize} on the levels above $m$ and using the truncated intervals after their assignment to level $m$ (see Algorithm \ref{algo:assign}). 
	Hence, $q$ is expected to be compared with $O(n/2^m)$ intervals in total and the cost of each such comparison is $O(1)$.
\end{proof}	
}

\rev{
\subsection{Setting $m$}
\label{sec:hierarchical:m}
As shown in Section \ref{sec:hint:analysis}, the space requirements
and the search performance of HINT$^m$ 
depend on the value of $m$.
For large values of $m$,
the cost of accessing
comparison-free results will dominate the computational cost of comparisons. 
This section presents an analytical study for estimating $m_{opt}$: the
smallest value of $m$, which is expected
to result in a HINT$^m$ of search performance close to the best
possible, while achieving the lowest possible space requirements.
Our study uses simple statistics \eat{about the input data and the query
workload;} namely, the number of intervals $n = |\mathcal{S}|$,
the mean length $\lambda_s$ of data intervals
and the mean length $\lambda_q$ of query
intervals\eat{ in the workload}.
We assume that the endpoints and the lengths of both
intervals and queries are uniformly distributed.

The overall cost of query evaluation consists of (1) the cost for
determining the relevant partitions per level, denoted by $C_p$, (2) the cost of
conducting comparisons between data intervals and the query,
denoted by $C_{cmp}$, and (3) the  
cost of accessing query results in the partitions for which we do not have to conduct
comparisons, denoted by $C_{acc}$. Cost $C_p$ is negligible, as the
partitions are determined by a small number $m$ of bit-shifting
operations. To estimate $C_{cmp}$, we need to estimate the number of
intervals in the partitions whereat we need to conduct comparisons
and multiply this by
the expected cost $\beta_{cmp}$ per comparison.
To estimate \eat{cost }$C_{acc}$, we need to estimate the number of
intervals in the corresponding partitions and multiply this by
the expected cost $\beta_{acc}$ of (sequentially) accessing and
reporting one interval.
$\beta_{cmp}$ and $\beta_{acc}$ are machine-dependent and can easily be
estimated by experimentation.

According to Algorithm \ref{algo:rangeqhier}, unless $\lambda_q$ is 
smaller than the length of a partition at level $m$, there will be two
partitions that require comparisons at level $m$, one partition at
level $m-1$, etc. with the expected number of partitions being at most
four (see Lemma~\ref{lem:expectedcompparts}).
Hence, we can assume that $C_{cmp}$ is practically dominated by the cost of processing
two partitions at the lowest level $m$.
As each partition at level $m$ is expected to have $n/2^m$
intervals (see Lemma~\ref{lem:lastlevelsize}), we have $C_{cmp} = \beta_{cmp}\cdot n/2^m$.
Then, the number of accessed intervals for which we expect to apply no comparisons is
$|Q|-2\cdot n/2^m$, where $|Q|$ is the total number of expected query
results. Under this, we have
$C_{acc} = \beta_{acc}\cdot(|Q|-2\cdot
n/2^m)$.
We can estimate $|Q|$ using the selectivity analysis for
(multidimensional) intervals and range queries in \cite{PagelSTW93} as
$|Q| = n\cdot \frac{\lambda_s + \lambda_q}
{\Lambda}$, where $\Lambda$ is the length of the entire domain \eat{which includes}with
all intervals in $\mathcal{S}$ (i.e., $\Lambda = \max_{\forall s\in \mathcal{S}}s.end-\min_{\forall s\in \mathcal{S}}s.st$).

With $C_{cmp}$ and $C_{acc}$\eat{ defined}, we now discuss how to estimate $m_{opt}$. First, we
%
gradually increase $m$ from $1$ up to its max value $m'$
(determined by $\Lambda$), and compute the expected cost
$C_{cmp}+C_{acc}$.
For $m = m'$, HINT$^m$ \eat{essentially }corresponds to the comparison-free HINT with the lowest expected cost. Then, we select as $m_{opt}$ the lowest value of $m$ for which $C_{cmp}+C_{acc}$ converges to the cost of the $m = m'$ case.

\eat{
old stuff:

methodology:
\begin{itemize}
\item assume a query selectivity formula based on uniform assumptions
  by adapting spatial query selectivity estimators and estimate number
  of results based on $\rho$. 
\item assume a specific cost per comparison and a specific cost for
  adding each resulting interval to the result set (XOR in our
  implementation).
\item estimate the total cost of comparisons as 4*$\log$*leafbucketsize*compcost 
\item estimate the total cost of comparison-free result additions as
  numberofresults*costofadding
\item simulate query cost for increasing values of $m$ until
  convergence, then choose $m_{opt}$ as convergence value
\end{itemize}

}
}

\rev{
\subsection{Updates}
\label{sec:hierarchical:updates}

We handle insertions to an existing HINT or HINT$^m$ index by calling Algorithm~\ref{algo:assign} for each new interval $s$. Small adjustments are needed for HINT$^m$ to add $s$ to the originals division at the first partition assignment, i.e., to $P^O_{\ell,a}$ or $P^O_{\ell,b}$, and to the replicas division for every other partition, i.e., to $P^R_{\ell,a}$ or $P^R_{\ell,b}$
Finally, we handle deletions using tombstones, similarly to previous studies \cite{Lomet75, Overmars83} and recent indexing approaches \cite{FerraginaV20}. Given an interval $s$ for deletion, we first search the index to locate all partitions that contain $s$ (both as original and as replica) and then, replace the id of $s$ by a special ``tombstone'' id, which signals the logical deletion.
}

\section{Optimizing HINT$^m$}
\label{sec:opts}
In this section,
we discuss optimization techniques, which greatly improve the performance
of HINT$^m$ (and HINT) in practice.
First, we show how to reduce the number of partitions in HINT$^m$ where 
comparisons are performed
and how to avoid
accessing unnecessary data.
Next, we show how to handle very sparse or skewed data at each
level of HINT/HINT$^m$. Another (orthogonal) optimization is
decoupling the storage of the interval ids with the storage of
interval endpoints in each partition. 
\rev{
Finally, we revisit updates under the prism of these optimizations.
}

\subsection{Subdivisions and space decomposition}
\label{sec:opts:comp}
Recall that, at each level $\ell$ of HINT$^m$,
every partition $P_{\ell,i}$ is divided into 
$P^O_{\ell,i}$ (holding originals) and 
$P^R_{\ell,i}$ (holding replicas).
We propose to further divide each \eat{set }$P^O_{\ell,i}$ into
$P^{O_{in}}_{\ell,i}$ and
$P^{O_{aft}}_{\ell,i}$, so that
$P^{O_{in}}_{\ell,i}$ (resp. $P^{O_{aft}}_{\ell,i}$) holds the
intervals from
$P^{O_{in}}_{\ell,i}$
that end {\em inside} (resp. {\em after})
partition $P_{\ell,i}$.
Similarly, each $P_{\ell,i}^R$ is divided into $P_{\ell,i}^{R_{in}}$ and $P_{\ell,i}^{R_{aft}}$.

\stitle{Range queries that overlap with multiple partitions.}
Consider a range query $q$,
which overlaps with a sequence of {\em more than one partitions} at
level $\ell$.
As already discussed,
if we have to conduct comparisons in the first such partition $P_{\ell,f}$,
we should do so for all intervals in $P_{\ell,f}^O$ and $P_{\ell,f}^R$.
By subdividing $P_{\ell,f}^O$ and $P_{\ell,f}^R$, we get the following lemma:
\begin{lemma}\label{lem:pf}
  If $P_{\ell,f}\ne P_{\ell,l}$ (1) each interval $s$ in $P_{\ell,f}^{O_{in}} \cup P_{\ell,f}^{R_{in}}$ overlaps with $q$ iff $s.end\ge q.st$; and
 (2)  all intervals $s$ in $P_{\ell,f}^{O_{aft}}$ and $P_{\ell,f}^{R_{aft}}$ are guaranteed to overlap with $q$.
\end{lemma}
\rev{
\begin{proof}
	Follows directly from the fact that $q$ starts {\em inside} $P_{\ell,f}$ but  ends {\em after} $P_{\ell,f}$.
\end{proof}	
}
Hence, we need {\em just one comparison} for each interval in
$P_{\ell,f}^{O_{in}} \cup P_{\ell,f}^{R_{in}}$, whereas we can report all intervals
$P_{\ell,f}^{O_{aft}}\cup P_{\ell,f}^{R_{aft}}$ as query results {\em without any
  comparisons}.
As already discussed, for all
partitions $P_{\ell,i}$  between $P_{\ell,f}$ and  $P_{\ell,l}$, we just report intervals in
$P_{\ell,i}^{O_{in}} \cup P_{\ell,i}^{O_{aft}}$ as results, without any comparisons,
whereas for the last partition $P_{\ell,l}$, we perform {\em one comparison}
  per interval
in $P_{\ell,l}^{O_{in}} \cup P_{\ell,l}^{O_{aft}}$.

\stitle{Range queries that overlap with a single partition.}
If the range query $q$ overlaps only
one partition $P_{\ell,f}$ at level $\ell$, 
we can use following lemma to minimize the
necessary comparisons:
\begin{lemma}\label{lem:1p}
If $P_{\ell,f} = P_{\ell,l}$ then
\begin{itemize}
\item each interval $s$ in $P_{\ell,f}^{O_{in}}$ overlaps with $q$ iff $s.st \le q.end \wedge q.st \leq s.end$,
\item each interval $s$ in $P_{\ell,f}^{O_{aft}}$ overlaps with $q$ iff $s.st \le q.end$,
\item each interval $s$ in $P_{\ell,f}^{R_{in}}$ overlaps with $q$ iff $s.end \ge q.st$,
\item all intervals in $P_{\ell,f}^{R_{aft}}$ overlap with $q$.
\end{itemize}
\end{lemma}
\rev{
\begin{proof}
	All intervals $s\in P_{\ell,f}^{O_{aft}}$ end after $q$, so $s.st \le q.end$ suffices as an overlap test.  
	All intervals $s\in P_{\ell,f}^{R_{in}}$ start before $q$, so $s.st \le q.end$ suffices as an overlap test.
	All intervals $s\in P_{\ell,f}^{R_{aft}}$ start before and end after $q$, so they are guaranteed results.
\end{proof}	
}
\begin{figure}[t]
     \includegraphics[width=0.99\columnwidth]{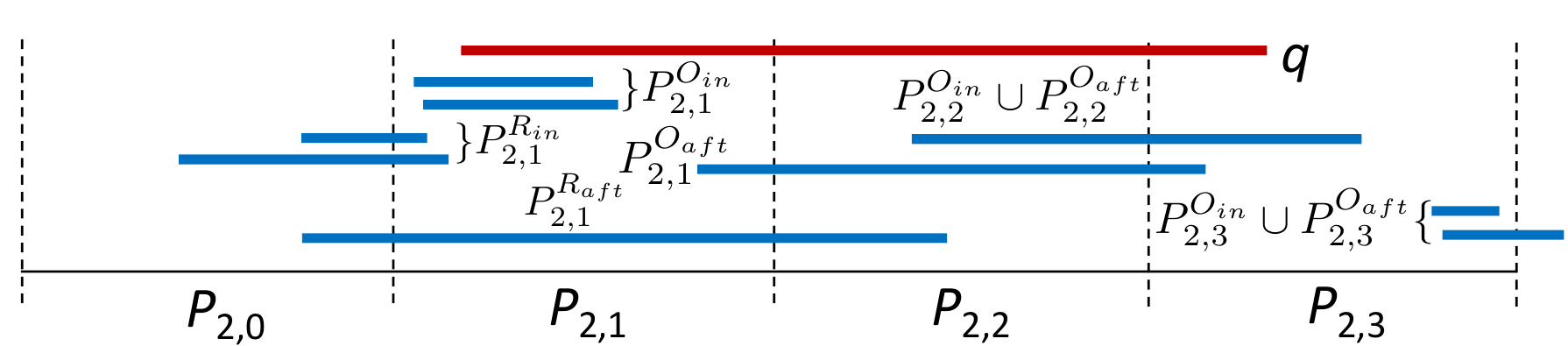}
     \vspace{-2ex}
  \caption{Partition subdivisions in HINT$^m$ (level $\ell=2$)}
  \label{fig:flatplusdecomp}
  \vspace{-3ex}
\end{figure}
Overall, the subdivisions help us to minimize the number of intervals
in each partition, for which we have to apply comparisons.
Figure~\ref{fig:flatplusdecomp} shows the subdivisions which are accessed by query $q$
at level $\ell=2$ of a HINT$^m$ index.
In partition $P_{\ell,f}=P_{2,1}$, all four subdivisions are accessed, but comparisons are needed only for intervals in $P_{2,1}^{O_{in}}$ and  $P_{2,1}^{R_{in}}$.
In partition $P_{2,2}$, only the originals (in  
$P_{2,2}^{O_{in}}$ and $P_{2,2}^{O_{aft}}$) are accessed and
reported without any comparisons.
Finally, in $P_{\ell,l}=P_{2,3}$, only the originals (in  
$P_{2,3}^{O_{in}}$ and $P_{2,3}^{O_{aft}}$) are accessed and compared to $q$.

\subsubsection{Sorting the intervals in each subdivision}
\label{sec:flat:plus:sorting}
We can keep the intervals in each subdivision sorted,
in order to reduce the number of comparisons for queries that access them.
For example, let us examine the last partition $P_{\ell,l}$ that
overlaps with a query $q$ at a level $\ell$.
If the intervals $s$ in $P_{\ell,l}^{O_{in}}$ are sorted \eat{based }on their start endpoint (i.e., $s.st$), we can simply access and report the intervals until the first $s\in  P_{\ell,l}^{O_{in}}$, such that $s.st>q.end$. Or, we can perform binary search to find the first  $s\in  P_{\ell,l}^{O_{in}}$, such that $s.st> q.end$ and then scan and report all intervals before $s$.
Table~\ref{tab:orderings} (second column) summarizes the sort orders for each of the four subdivisions of a partition that can be beneficial in range query evaluation. 
For  a subdivision $P_{\ell,i}^{O_{in}}$, intervals may have to be compared based on their start point (if  $P_{\ell,i}=P_{\ell,f}$), or based on their end point (if  $P_{\ell,i}=P_{\ell,l}$), or based on both points (if $P_{\ell,i}=P_{\ell,f}=P_{\ell,l}$). Hence, we choose to sort based on either $s.st$ or $s.end$ to accommodate two of these three cases.
For  a subdivision $P_{\ell,i}^{O_{aft}}$, intervals may only have to be compared based on their start point (if $P_{\ell,i}=P_{\ell,l}$).  
For  a subdivision $P_{\ell,i}^{R_{in}}$, intervals may only have to be compared based on their end point (if $P_{\ell,i}=P_{\ell,f}$).  
Last, for  a subdivision $P_{\ell,i}^{R_{aft}}$, there is never any need
to compare the intervals,
so, no order provides any search benefit.

\begin{table}
  \caption{Sort orders that can be beneficial}\label{tab:orderings}
	\footnotesize
     \vspace{-2ex}
  \begin{tabular}{|l||l|l|}
    \hline
    {\bf subdivision} & {\bf beneficial sorting} & {\bf necessary data}\\
    \hline\hline
    $P_{\ell,i}^{O_{in}}$& by  $s.st$ or by  $s.end$ & $s.id, s.st, s.end$\\
\rule{0pt}{3ex}    $P_{\ell,i}^{O_{aft}}$& by  $s.st$ & $s.id, s.st$\\
 \rule{0pt}{3ex}   $P_{\ell,i}^{R_{in}}$& by  $s.end$ & $s.id, s.end$\\
 \rule{0pt}{4ex}   $P_{\ell,i}^{R_{aft}}$& no sorting & $s.id$\\ \hline  
  \end{tabular}
\vspace{-3ex}
\end{table}

\subsubsection{Storage optimization}
\label{sec:flat:plus:decomposition}
So far, we have assumed that each interval $s$ is stored in the partitions whereto $s$ is assigned as a triplet $\langle s.id, s.st, s.end \rangle$.
However, if we split the partitions into subdivisions,
we do not need to keep all information
of the intervals in them. Specifically,
for each subdivision $P_{\ell,i}^{O_{in}}$, we may need to use
$s.st$ and/or $s.end$ for each interval $s\in P_{\ell,i}^{O_{in}}$, while for each subdivision $P_{\ell,i}^{O_{aft}}$, we may need to use $s.st$ for each $s\in P_{\ell,i}^{O_{in}}$, but we will never need $s.end$. From the intervals $s$ of each  subdivision $P_{\ell,i}^{R_{in}}$, we may need $s.end$, but we will never use $s.st$. Finally, for each  subdivision $P_{\ell,i}^{R_{in}}$, we just have to keep the $s.id$ identifiers of the intervals.
Table~\ref{tab:orderings} (third column) summarizes the data that we need to keep from each interval in the subdivisions of each partition.
Since each interval $s$ is stored as original just once in the
entire index, but as replica in possibly multiple partitions, 
space can be saved by storing only the necessary data, especially 
if the
intervals span multiple partitions.
Last, note that even when we 
do not apply the subdivisions, but just use \eat{two }divisions
$P_{\ell,i}^{O}$ and $P_{\ell,i}^{R}$ (as suggested in Section~\ref{sec:hierarchical:partial}), we do not have to store
the start points $s.st$ of all intervals in $P_{\ell,i}^{R}$\eat{ divisions},
since they are never used in comparisons.

\subsection{Handling data skewness and sparsity}
\label{sec:opts:skew}
Data skewness and sparsity may cause many partitions to be empty,
especially at the lowest levels of HINT (i.e., large values of
$\ell$).
\eat{At these levels,
there could be many empty partitions.}
Recall that a query accesses a sequence of multiple $P^O_{\ell,i}$ partitions
at each level $\ell$.
Since the intervals are physically distributed in the partitions,
this results into the unnecessary accessing of empty partitions
and may cause cache misses.
We propose a storage organization where all $P_{\ell,i}^{O}$ divisions at the
same level $\ell$ are merged into a single table $T_{\ell}^{O}$ and an
auxiliary index is used to find each non-empty division.%
\footnote{For simplicity, we discuss this organization when a partition $P_{\ell,i}$ is divided into $P^O_{\ell,i}$ and $P^R_{\ell,i}$; the same idea can be straightforwardly applied also when the four subdivisions discussed in Section~\ref{sec:flat:plus:decomposition} are used.}
The auxiliary index locates the first non-empty partition,
which is greater than or equal to the $\ell$-prefix of $q.st$ (i.e.,
via binary search or a binary search tree).
From thereon, the nonempty partitions which overlap with the query interval
are accessed sequentially and distinguished with the help of the
auxiliary index.
Hence, the contents of the relevant $P^O_{\ell,i}$'s to each query are
always accessed sequentially.
Figure~\ref{fig:onearraysparse}(a) shows an example at
level $\ell=4$ of HINT$^m$.  From the total $2^\ell=16$ $P^O$ partitions at that
level, only 5 are nonempty (shown in grey at the top of the
figure): $P^O_{4,1}, P^O_{4,5}, P^O_{4,6}, P^O_{4,8}, P^O_{4,13}$.
All 9 intervals in them (sorted by start point) are unified in a
single table $T_4^O$ as shown at the bottom of the figure
(the binary representations of the interval endpoints are shown).
At the moment, ignore the ids column for $T_4^O$\eat{ shown} at the
\eat{bottom-}right of the figure.
The sparse index for $T_4^O$ has one entry per nonempty partition
pointing to the first interval in it. For the query \eat{shown }in the
example, the index is used to find the first nonempty partition
$P^O_{4,5}$, for which the id is greater than or equal to the $4$-bit
prefix $0100$ of $q.st$. All relevant non-empty partitions 
$P^O_{4,5}, P^O_{4,6}, P^O_{4,8}$ are accessed sequentially from
$T_4^O$, until the position of the first interval of $P^O_{4,13}$.

Searching for the first partition $P^O_{\ell,f}$
that overlaps with $q$ at each level can be quite expensive 
when numerous nonempty partitions exist.
To alleviate this issue, we suggest adding to the auxiliary index, a link from each
partition $P^O_{\ell,i}$ to the partition $P^O_{\ell-1,j}$ at the
level above, such that $j$ is the smallest number greater than or
equal to $i\div 2$, for which partition $P^O_{\ell-1,j}$ is not
empty. Hence, instead of performing binary search at level $\ell-1$,
we use the link from the first partition $P^O_{\ell,f}$ relevant to the query at
level $\ell$ and (if necessary)
apply a linear search backwards starting from the pointed partition
$P^O_{\ell-1,j}$ to identify the first
non-empty partition $P^O_{\ell-1,f}$ that overlaps with $q$.
Figure~\ref{fig:onearraysparse}(b) shows an example, where each
nonempty partition at level $\ell$ is linked with the first nonempty partition with
greater than or equal prefix at the level $\ell-1$ above. Given 
query example $q$,
we use the auxiliary index to find the first nonempty partition
$P^O_{4,5}$ which overlaps with $q$ and also sequentially access  
$P^O_{4,6}$ and $P^O_{4,8}$. Then, we follow the pointer from
$P^O_{4,5}$ to $P^O_{3,4}$ to find the first nonempty partition at
level $3$, which overlaps with $q$. We repeat this to get partition
$P^O_{2,3}$ at level $2$, which however is not guaranteed to be the
first one  
overlapping
with $q$, so we go backwards to \eat{reach }$P^O_{2,3}$.

\eat{
\begin{figure}
\begin{tabular}{cc}
\hspace{-2ex}\includegraphics[width=0.5\columnwidth]{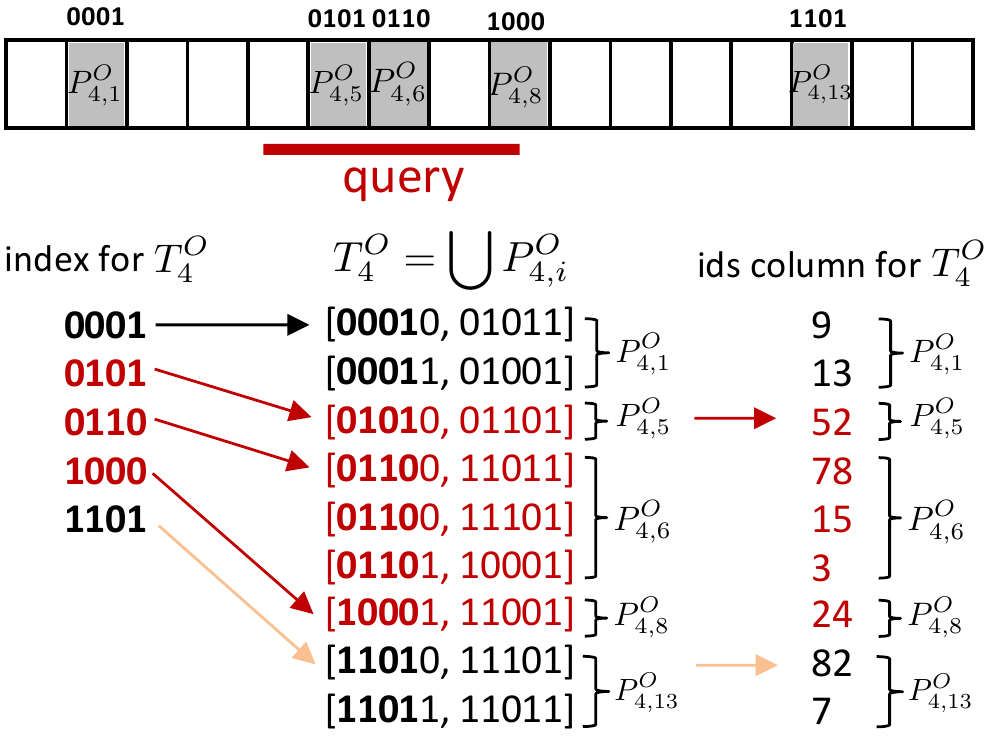}
&\hspace{-1ex}\includegraphics[width=0.5\columnwidth]{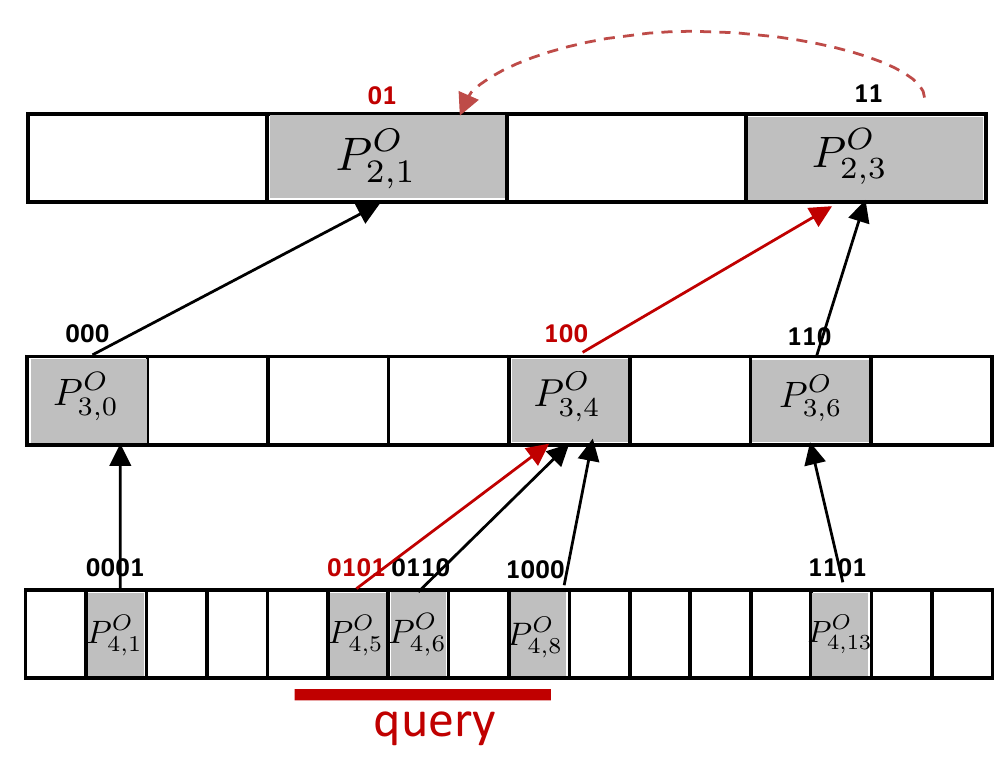}\\
  (a) auxiliary index&
                       (b) linking between levels
\end{tabular}
\caption{Storage and indexing optimizations}
\label{fig:onearraysparse}
\end{figure}
}
\begin{figure}
\centering
\includegraphics[width=0.7\columnwidth]{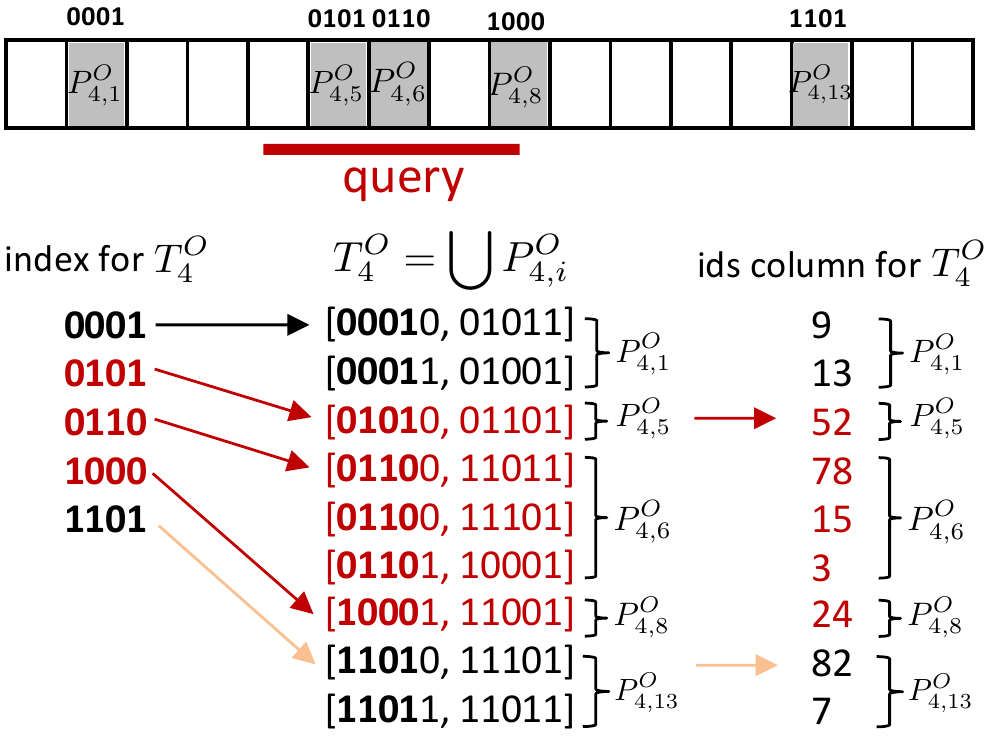}\\
(a) auxiliary index\\\vspace{2ex}
\vspace{-2mm}
\includegraphics[width=0.6\columnwidth]{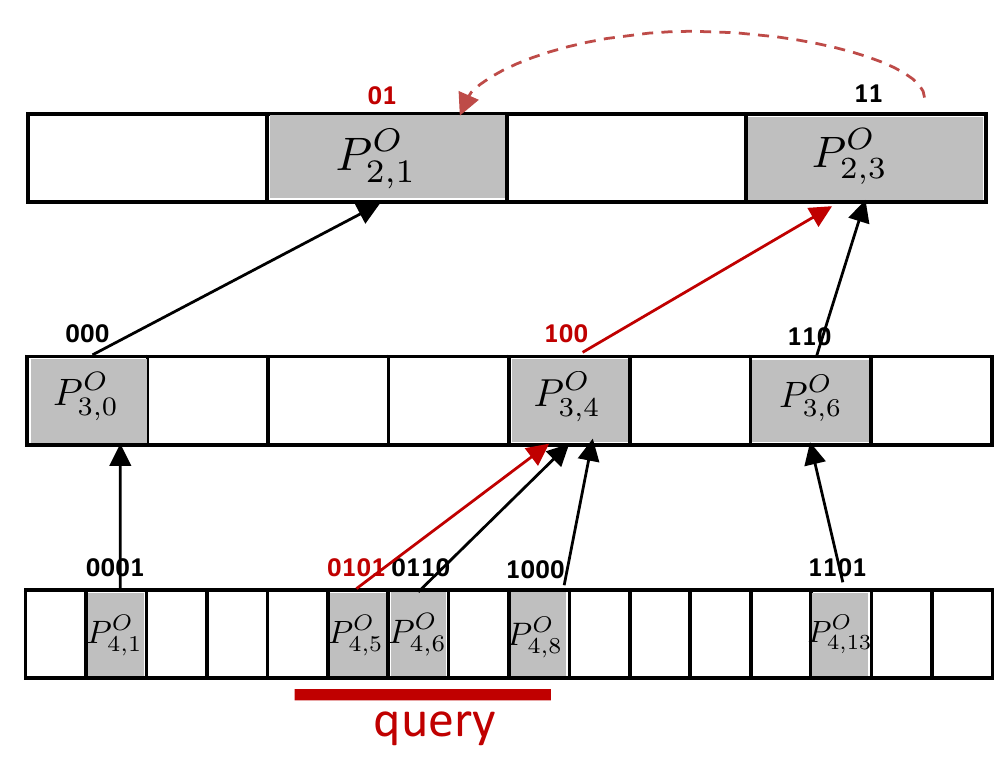}\\
\vspace{-1mm}
 (b) linking between levels
\vspace{-2mm}
\caption{Storage and indexing optimizations}
\label{fig:onearraysparse}
\vspace{-2ex}
\end{figure}

\subsection{Reducing cache misses}
\label{sec:opts:cache}
At most levels of HINT$^m$, no comparisons are conducted
and the only operations are processing the interval ids which qualify
the query.
In addition, even for the levels $\ell$ where comparisons are required, these
are only restricted to the first and the last partitions
$P^O_{\ell,f}$ and $P^O_{\ell,l}$ that overlap with $q$
and no comparisons are needed for the partitions that are in-between.
Summing up, when accessing any (sub-)partition for which no comparison
is required, we do not need any information about the
intervals, except for their ids. Hence, in our implementation, for
each (sub-)partition, we store the ids of all intervals in it in a
dedicated array (the {\em ids column})
and the interval endpoints (wherever necessary) in a
different array.%
\footnote{Similar to the previous section, this storage optimization can be straightforwardly employed also when a partition is divided into  $P_{\ell,i}^{O_{in}}$, $P_{\ell,i}^{O_{aft}}$, $P_{\ell,i}^{R_{in}}$, $P_{\ell,i}^{R_{aft}}$.}
If we need the id of an interval
that qualifies a comparison, we can access the
corresponding position of the ids column.
This storage organization greatly improves search performance by
reducing the cache misses, because
for the intervals that do not require comparisons, we only access
their ids and not their interval endpoints.
This optimization is orthogonal to and applied in combination with the
strategy discussed in Section~\ref{sec:opts:skew}, i.e., we store all
$P^O$ divisions at each level $\ell$ in a single table $T_{\ell}^{O}$,
which is decomposed to a column that stores the ids and another
table for the endpoint data of the intervals.
An example of the ids column is shown in Figure~\ref{fig:onearraysparse}(a). If, for a sequence of partitions at a
level, we do not have to perform any comparisons, we just access the
sequence of the interval ids that are part of the answer, which is
implied by the position of the first such partition (obtained
via the auxiliary index). In this example, all intervals in
$P^O_{4,5}$ and  $P^O_{4,6}$ are guaranteed to be query results
without any comparisons and they can be sequentially accessed from the
ids column without having to access the endpoints of the
intervals. The auxiliary index guides the search by identifying and
distinguishing between partitions for which comparisons should be
conducted (e.g.,  $P^O_{4,8}$) and those for which they are not
necessary.

\rev{
\subsection{Updates}
\label{sec:opts:updates}
A version of HINT$^m$ that uses \emph{all} techniques from Sections~\ref{sec:opts:comp}-\ref{sec:opts:skew}, is optimized for query operations. Under this premise, the index cannot efficiently support individual updates, i.e., new intervals inserted one-by-one. Dealing with updates in {\em batches} will be a better fit. This is a common practice for other update-unfriendly indices, e.g., the inverted index in IR\eat{information retrieval}. Yet, for mixed workloads (i.e., with both queries and updates), we adopt a hybrid setting where a
\emph{delta} index is maintained to digest the latest updates as discussed in Section~\ref{sec:hierarchical:updates},%
\footnote{\rev{Small adjustments are applied for the $P^{O_{in}}_{l,i}$, $P^{O_{aft}}_{l,i}$, $P^{R_{in}}_{l,i}$, $P^{R_{aft}}_{l,i}$ subdivisions and the storage optimizations.}} and a fully optimized HINT$^m$, which is updated periodically in batches, holds older data 
supporting deletions with tombstones. Both indices are probed when a query is evaluated.
}

\section{Experimental Analysis}
\label{sec:exps}
We compared our hierarchical index, detailed
in Sections~\ref{sec:hierarchical} \rev{and \ref{sec:opts}} against
the interval tree \cite{Edels80} (code from \cite{itree}),
the timeline index \cite{KaufmannMVFKFM13},
the (adaptive) period index \cite{BehrendDGSVRK19},
and a uniform 1D-grid.
All indices were implemented in \texttt{C++} and compiled
using \texttt{gcc} (v4.8.5) with \eat{flags }\texttt{-O3}\eat{, \texttt{-mavx}
and \texttt{-march=native}}.
\footnote{Source code available in \href{https://github.com/pbour/hint}{https://github.com/pbour/hint}.}
The tests ran on a dual Intel(R) Xeon(R) CPU E5-2630 v4 clocked at 2.20GHz with 384 GBs of RAM, running CentOS Linux\eat{ 8.2.2004}. 

\begin{table}
\centering
\caption{Characteristics of real datasets}
\label{tab:datasets}
\vspace{-1ex}
\footnotesize
\begin{tabular}{|l|@{~}c@{~}|@{~}c@{~}|@{~}c@{~}|@{~}c@{~}|}\hline
												&BOOKS							&WEBKIT							&TAXIS								&GREEND\\\hline\hline
Cardinality								&$2,\!312,\!602$ 				&$2,\!347,\!346$ 				&$172,\!668,\!003$			&$110,\!115,\!441$\\
%
\rev{Size [MBs]}						&\rev{$27.8$} 					&\rev{$28.2$} 					&\rev{$2072$}					&\rev{$1321$}\\
%
Domain [sec]							&$31,\!507,\!200$				&$461,\!829,\!284$			&$31,\!768,\!287$				&$283,\!356,\!410$\\
Min duration [sec]						&$1$									&$1$									&$1$									&$1$\\
Max duration [sec]					&$31,\!406,\!400$			&$461,\!815,\!512$			&$2,\!148,\!385$				&$59,\!468,\!008$\\
Avg. duration [sec] 					&$2,\!201,\!320$				&$33,\!206,\!300$			&$758$							&$15$\\
Avg. duration [\%] 					&6.98								&7.19								&0.0024							&0.000005\\
\hline
\end{tabular}
\vspace{-2ex}
\end{table}
\begin{table}
\centering
\caption{Parameters of synthetic datasets}
\label{tab:synthdatasets}
\vspace{-2mm}
\footnotesize
\begin{tabular}{|l|c|}\hline
\textbf{parameter} & \textbf{values} (defaults in {\bf bold})\\\hline\hline
Domain length
                                                                                                                        &32M,
                                                                                                                          64M,{\bf
                                                                                                                          128M},
  256M, 512M\\
  Cardinality
%

&
                                                                                                        {\bf 10M}, 50M, 100M, \rev{500M}, \rev{1B}\\
                                                                                                        $\alpha$ (interval length) 
                                                                                                                        &1.01,
  1.1, {\bf 1.2}, 1.4, 1.8\\
$\sigma$ (interval position) 
                                                                                                                        &
  10K, 100K, {\bf 1M}, 5M, 10M\\
\hline
\end{tabular}
\vspace{-2ex}
\end{table}

\subsection{Data and queries}
\label{sec:exps:setup}
\eat{For our analysis, w}We used 4 collections of real time intervals\eat{, which
have also been used in} from
previous works\eat{ \cite{DignosBG14,PiatovHD16,BourosM17,CafagnaB17,BourosMTT21}};
Table \ref{tab:datasets} summarizes their characteristics.
BOOKS \cite{BourosM17} contains the \eat{time intervals }periods during which books
were lent out by Aarhus \eat{public }libraries in 2013 (https://www.odaa.dk).
WEBKIT \cite{BourosM17,BourosM18,DignosBG14,PiatovHD16} records the file history
in the git repository of the Webkit project from 2001 to 2016
(https://webkit.org); the intervals indicate the periods during which a file did not change.
TAXIS \cite{BourosMTT21} includes the time periods of taxi trips (pick-up and drop-off
timestamps) from NYC 
(https://www1.nyc.gov/site/tlc/index.page) in 2013.
GREEND \cite{CafagnaB17,MonacchiEEDT14} records time periods of power
usage \eat{data }from households in Austria and Italy from January 2010 to
October 2014.
BOOKS and WEBKIT contain around 2M intervals each, which are quite
long on average;
\eat{while }TAXIS and GREEND contain over
100M relatively short intervals.

We also generated synthetic collections \eat{of intervals, in order }to
simulate different cases for the lengths and the skewness of the input
intervals.
Table \ref{tab:synthdatasets} shows the construction\eat{values of the} parameters \eat{that
determine}for the synthetic datasets and their default values. The domain
of the datasets 
ranges from 32M to 512M, which requires index level parameter $m$ to
range from $25$ to $29$ for a comparison-free HINT (similar to the real datasets).
The cardinality ranges from \rev{10M to 1B}.
The lengths of the intervals were generated using the
\texttt{random.zipf($\alpha$)} function in the \texttt{numpy} library. 
They follow a zipfian distribution according to the $p(x) = \frac{x^{-a}}{\zeta(a)}$ probability density function, where $\zeta$ is the Riemann Zeta function.
A small value of $\alpha$ results in most intervals being relatively
long,
while a large value results in the great majority of intervals having
length 1.
The positions of the {\em middle points} of the intervals are generated from
a normal distribution centered at the middle point $\mu$ of the domain.
Hence, the middle point of each interval is generated by calling
numpy's \texttt{random.normalvariate($\mu, \sigma$)}. The greater the value
of $\sigma$ the more spread the intervals are in the domain.

On the real datasets, we ran range queries \eat{which are }uniformly distributed
in the domain.
On the synthetic\eat{ datasets}, the positions of the queries follow the
distribution of the data. In both cases, the extent of the query
intervals were fixed to a percentage of the domain size (default
0.1\%). At each test\eat{experimental instance}, we ran 10K random queries, in
order to measure the overall throughput.
\rev{
Measuring query throughput instead of average time per query makes
sense in applications or services that manage huge volumes of
interval data and offer a search interface to billions of users
simultaneously (e.g., public historical databases).
}

\subsection{Optimizing HINT/HINT$^m$}
In our first set of experiments, we study the best setting for our hierarchical index. Specifically, we compare the effectiveness of the two query evaluation approaches discussed in Section~\ref{sec:hierarchical:qeval} and investigate the impact of the optimizations described in Section~\ref{sec:opts}.

\subsubsection{Query evaluation approaches on HINT$^m$}
\begin{figure}[t]
\begin{tabular}{cc}
BOOKS		&TAXIS\\
\hspace{-1ex}\includegraphics[width=0.46\columnwidth]{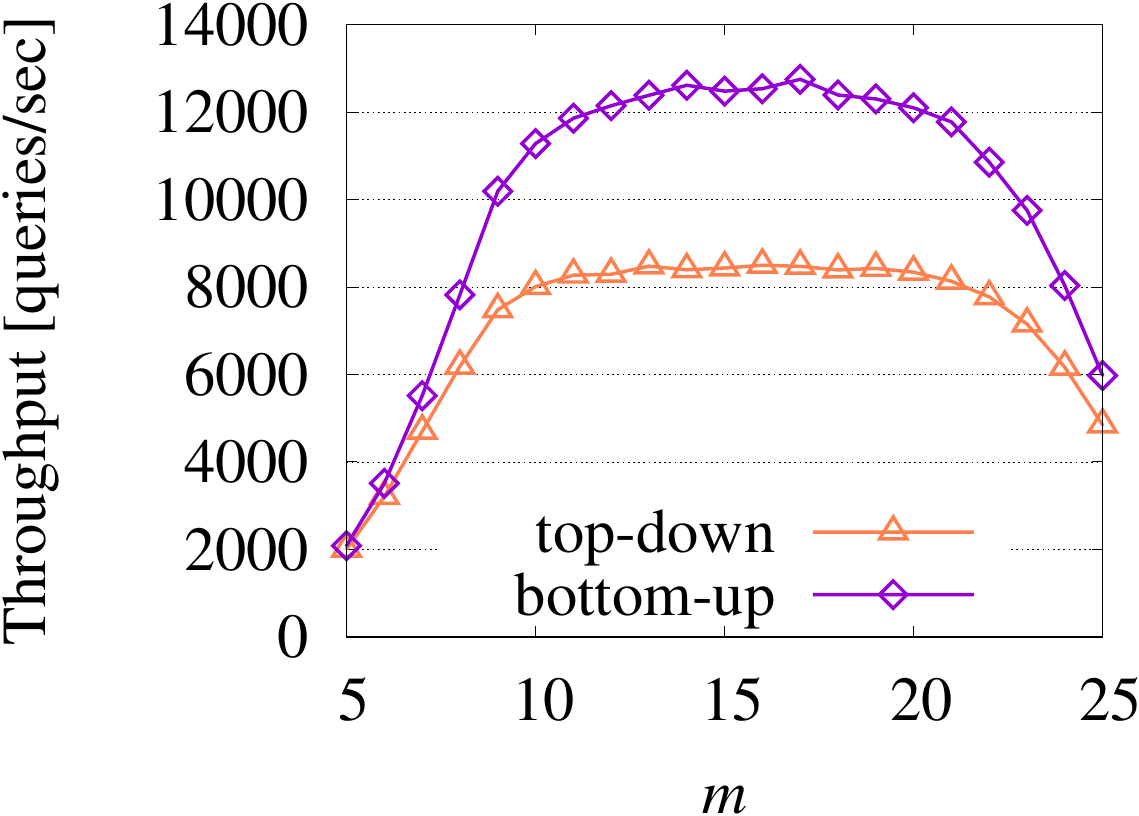}
&\hspace{-1ex}\includegraphics[width=0.46\columnwidth]{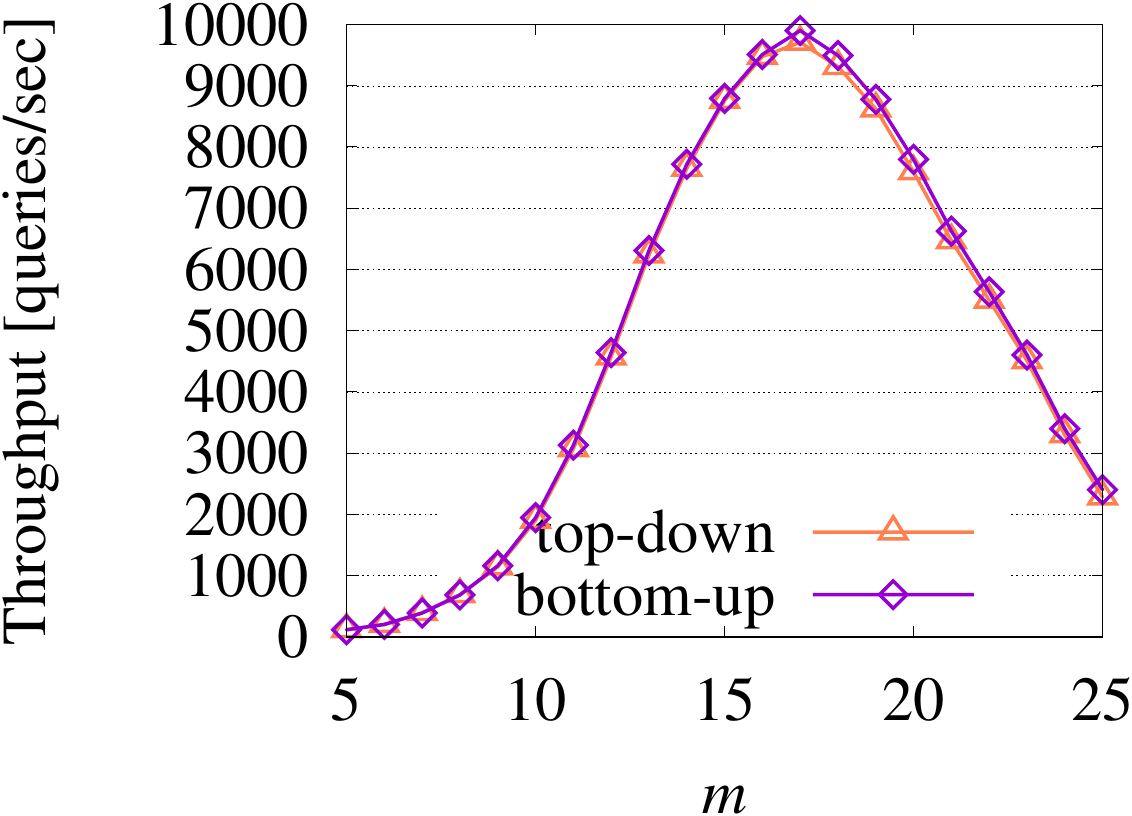}
\end{tabular}
\vspace{-3ex}
\caption{Optimizing HINT$^m$: query evaluation approaches}
\label{fig:hierm_qstrategy}
\vspace{-3ex}
\end{figure}
\eat{
We compare the straightforward \emph{top-down} approach for evaluating
range queries on HINT$^m$ that uses solely Lemma~\ref{lem:compcut},
against the \emph{bottom-up} illustrated in
Algorithm~\ref{algo:rangeqhier} which additionally employs Lemma~\ref{lem:firstlast}. 
Figure~\ref{fig:hierm_qstrategy} reports
the throughput of each approach on our real interval collections,
while varying the number of levels $m$ in the index. We observe that
the \emph{bottom-up} approach significantly outperforms
\emph{top-down} for BOOKS and WEBKIT while for TAXIS and GREEND, this
performance gap is very small. As expected, \emph{bottom-up} performs
at its best when the input\eat{ collection} contains long intervals which
are indexed on high levels of index; this is the case with BOOKS and
WEBKIT. In contrast, the intervals in TAXIS and GREEND are very short
and so, indexed at the bottom level of HINT$^m$, while the
majority of the partitions at the higher levels are empty. As a
result, \emph{top-down} conducts
no comparisons at
higher levels.
For the rest of our tests, HINT$^m$ uses the \emph{bottom-up} approach (i.e.,
Algorithm~\ref{algo:rangeqhier}).
}
We compare the straightforward \emph{top-down} approach for evaluating
range queries on HINT$^m$ that uses solely Lemma~\ref{lem:compcut},
against the \emph{bottom-up} illustrated in
Algorithm~\ref{algo:rangeqhier} which additionally employs Lemma~\ref{lem:firstlast}. 
Figure~\ref{fig:hierm_qstrategy} reports
the throughput of each approach on BOOKS and TAXIS,
while varying the number of levels $m$ in the index. 
Due to lack of space, we omit the results for WEBKIT and GREEND that follow exactly the same trend with BOOKS and TAXIS, respectively.
We observe that the \emph{bottom-up} approach significantly outperforms
\emph{top-down} for BOOKS while for TAXIS, this
performance gap is very small. As expected, \emph{bottom-up} performs
at its best for inputs\eat{ collection} that contain long intervals which
are indexed on high levels of index, i.e., the intervals in BOOKS. 
In contrast, the intervals in TAXIS are very short
and so, indexed at the bottom level of HINT$^m$, while the
majority of the partitions at the higher levels are empty. As a
result, \emph{top-down} conducts
no comparisons at
higher levels.
For the rest of our tests, HINT$^m$ uses the \emph{bottom-up} approach (i.e.,
Algorithm~\ref{algo:rangeqhier}).

\subsubsection{Subdivisions and space decomposition}
\begin{figure}[t]
\eat{
\begin{small}
\fbox{\parbox{220pt}
{
\begin{center}
{\footnotesize base}
\includegraphics[width=0.06\columnwidth]{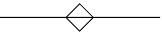}
\hspace{3ex}
{\footnotesize subs+sort}
\includegraphics[width=0.06\columnwidth]{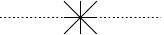}
{\footnotesize subs+sopt}
\includegraphics[width=0.06\columnwidth]{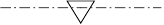}
\hspace{3.7ex}
{\footnotesize subs+sort+sopt}
\includegraphics[width=0.06\columnwidth]{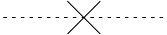}
\end{center}
}
}
\end{small}
}
\begin{tabular}{cc}
BOOKS &TAXIS\\
\hspace{-1ex}\includegraphics[width=0.46\columnwidth]{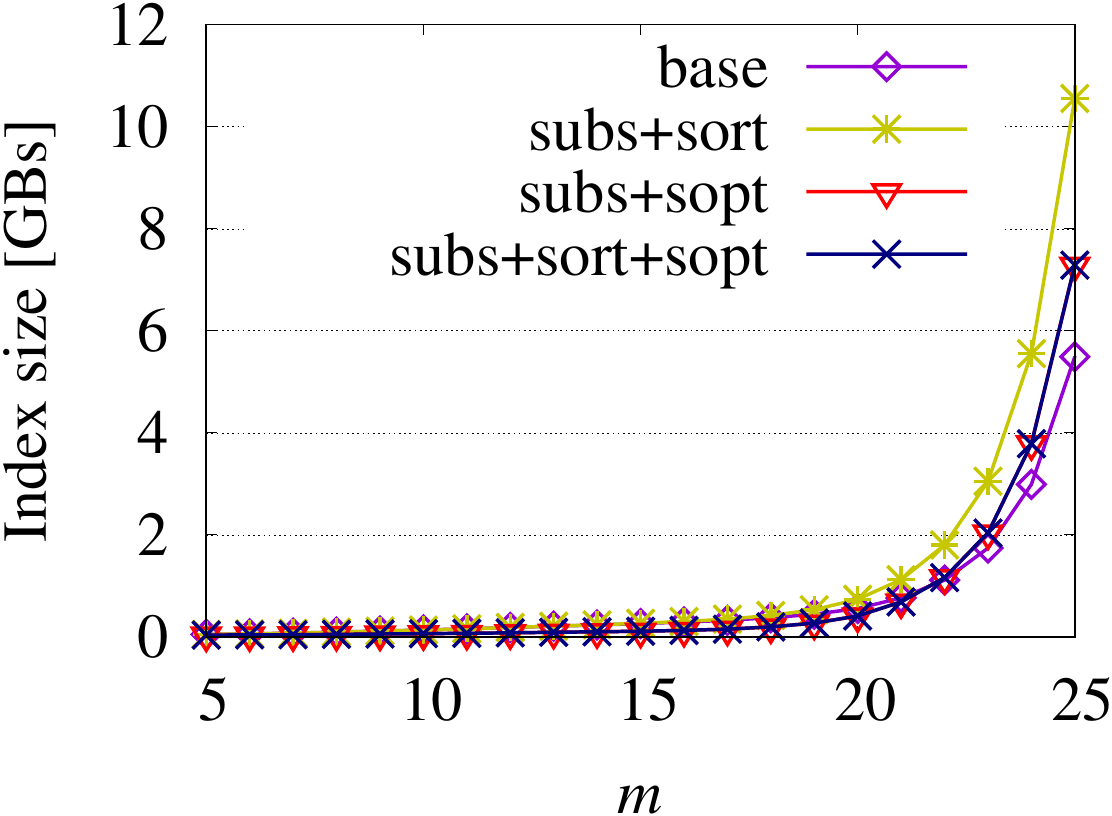}
&\includegraphics[width=0.46\columnwidth]{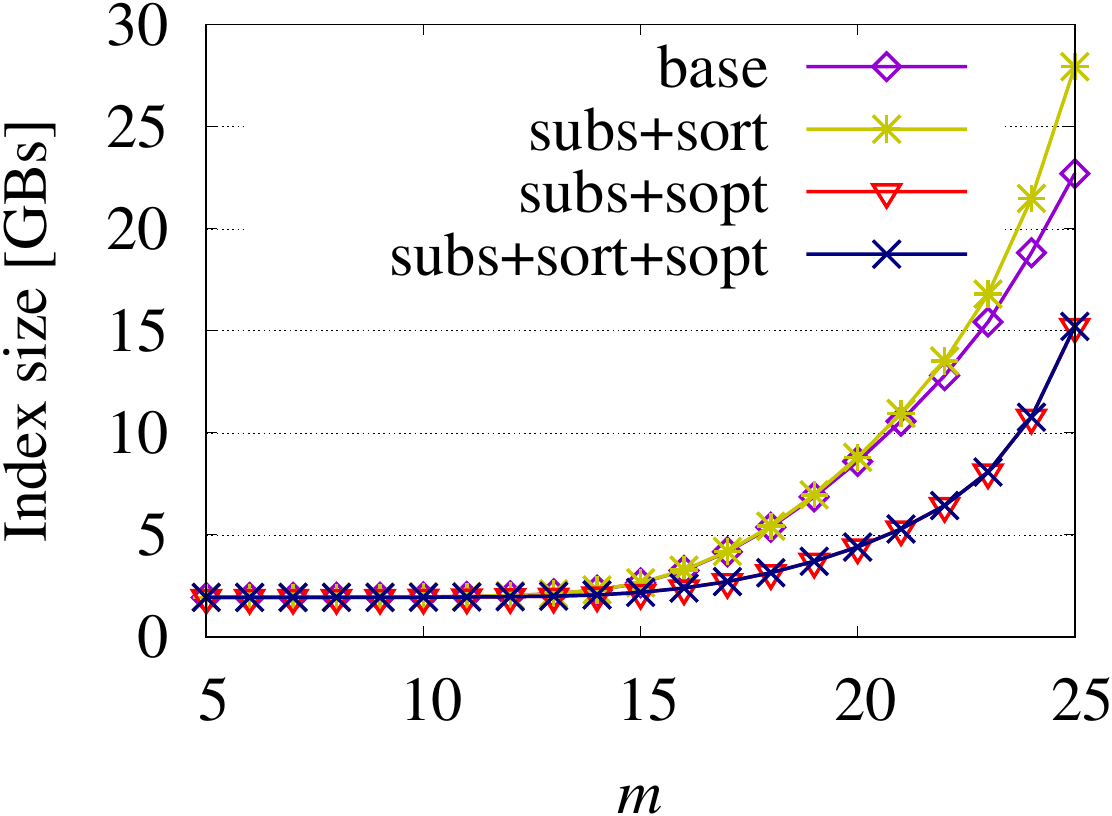}\\
\hspace{-1ex}\includegraphics[width=0.46\columnwidth]{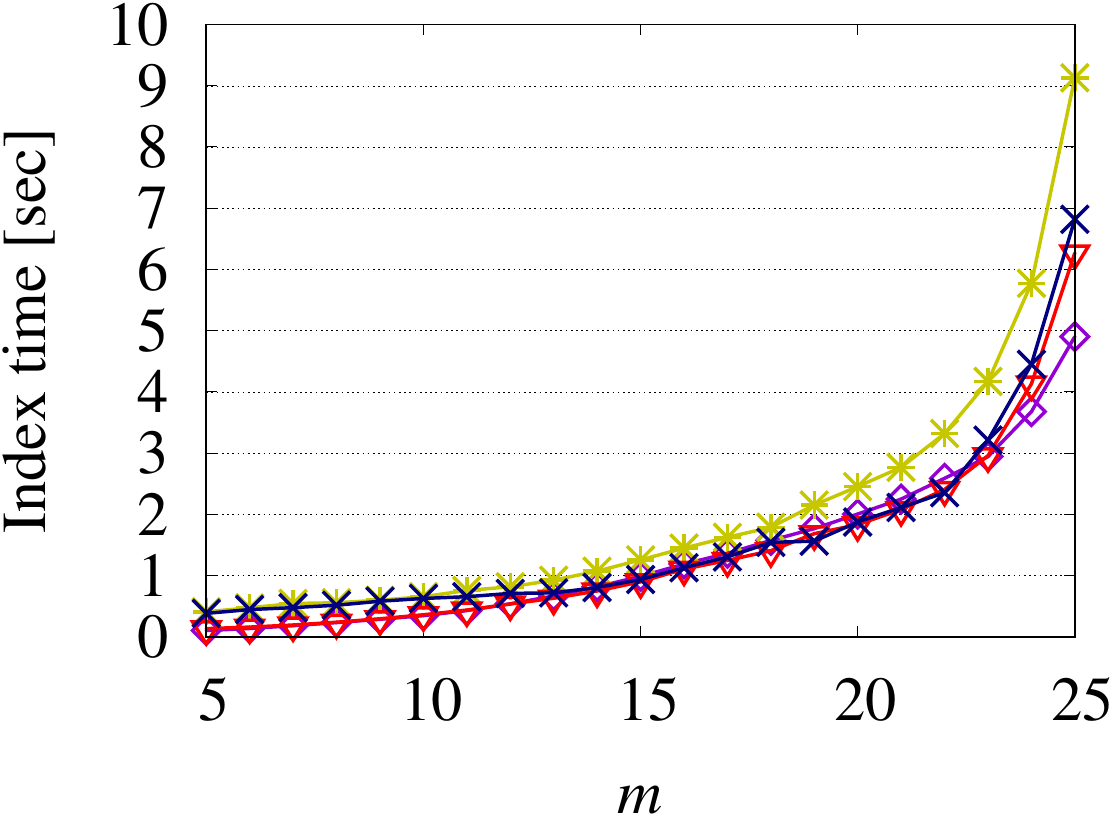}
&\includegraphics[width=0.46\columnwidth]{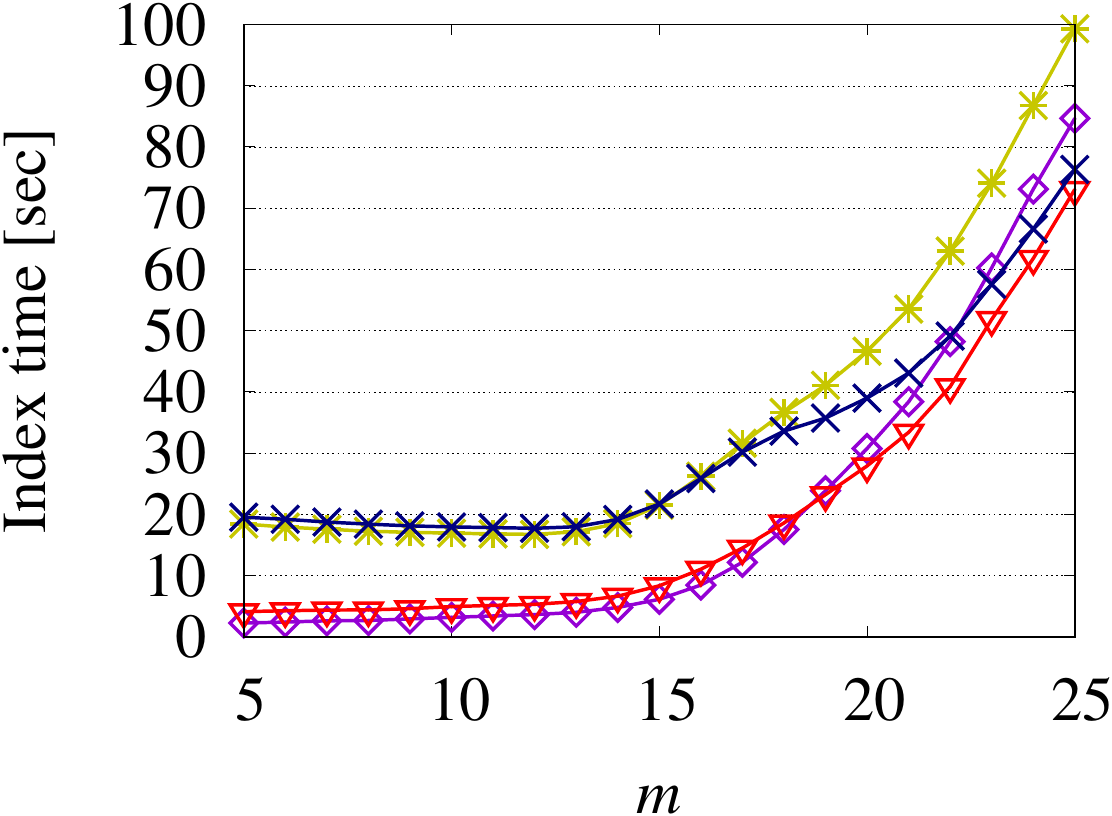}\\
\hspace{-1ex}\includegraphics[width=0.46\columnwidth]{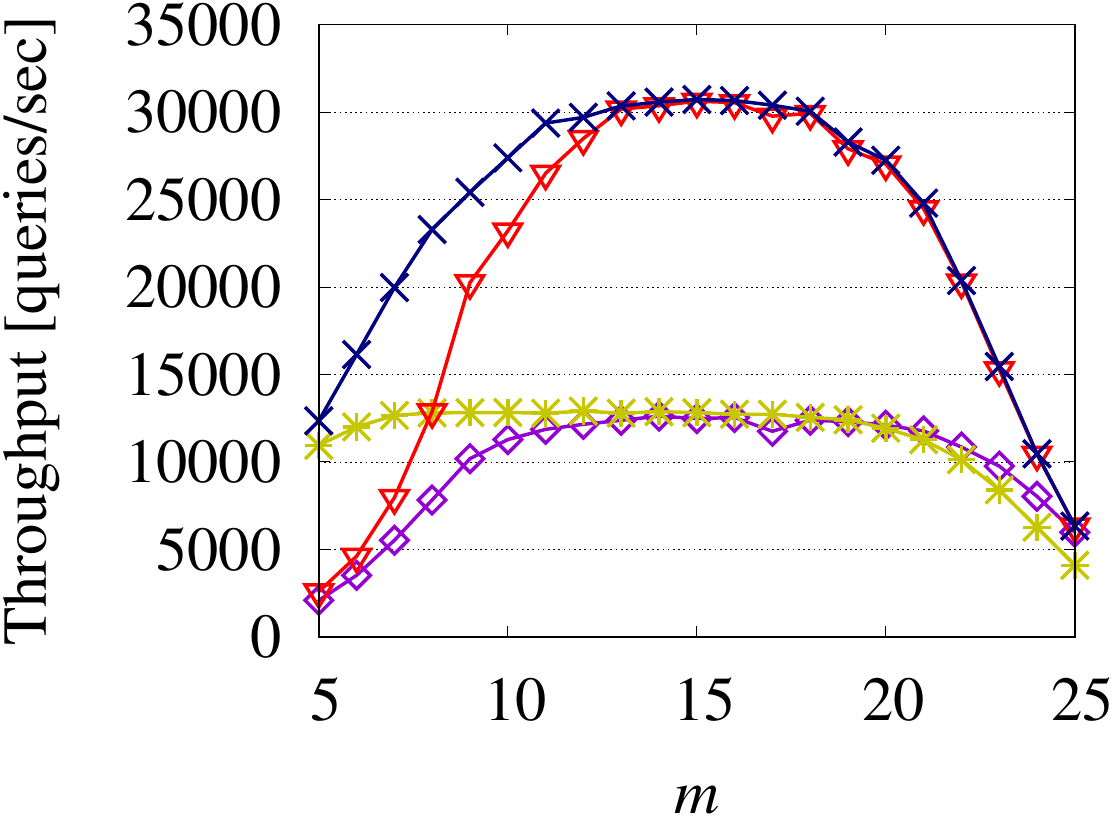}
&\includegraphics[width=0.46\columnwidth]{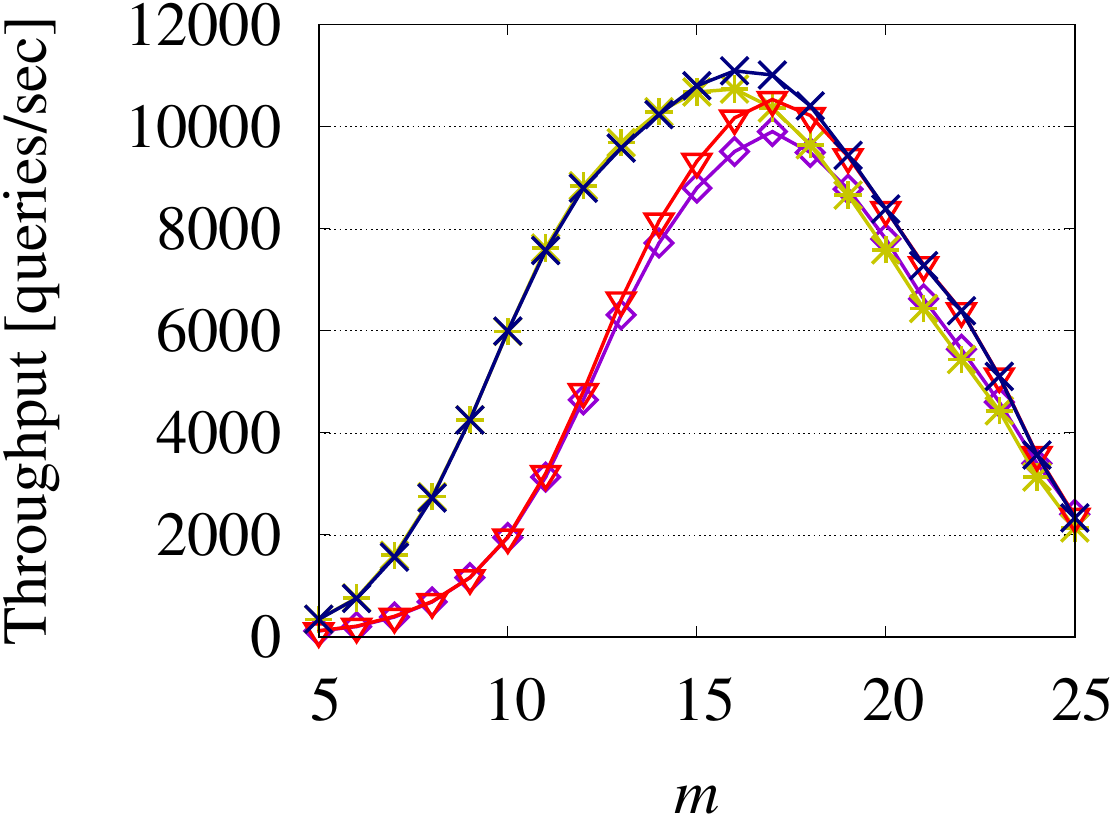}
\end{tabular}
\vspace{-3ex}
\caption{Optimizing HINT$^m$: subdivisions and space decomposition}
\label{fig:hierm_comps}
\vspace*{-3ex}
\end{figure}
We next evaluate the \emph{subdivisions} and \emph{space decomposition} optimizations described in Section~\ref{sec:opts:comp} for HINT$^m$. Note that these techniques are not applicable to our comparison-free HINT as the index stores only interval ids. Figure~\ref{fig:hierm_comps} shows the effect of the optimizations on BOOKS and TAXIS, for different values of $m$;
similar trends were observed in WEBKIT and GREEND, respectively.
The plots include (1) a \emph{base} version of HINT$^m$, which employs none of the proposed optimizations, (2) \emph{subs+sort+sopt}, with all optimizations activated, (3) \emph{subs+sort}, which only sorts the subdivisions (Section~\ref{sec:flat:plus:sorting}) and (iv) \emph{subs+sopt}, which uses only the storage optimization (Section~\ref{sec:flat:plus:decomposition}).
We observe that the \emph{subs+sort+sopt} version of HINT$^m$ is
superior to all three other versions, on all tests. Essentially, the
index benefits from the \emph{sub+sort} setting only when $m$ is
small, i.e., below 15, at the expense of increasing the index time
compared to \emph{base}. In this case, the partitions contain a large
number of intervals and therefore, using binary search or scanning
until the first interval that does not overlap the query range, will
save on the conducted comparisons. On the other hand, the
\emph{subs+sopt} optimization significantly reduces the space
requirements 
of the index. As a result, the version incurs a higher cache hit ratio and so, a higher throughput compared to \emph{base} is achieved, especially for large values of $m$, i.e., higher than 10. The \emph{subs+sort+sopt} version manages to combine the benefits of both \emph{subs+sort} and \emph{subs+sopt} versions, i.e., high throughput in all cases, with low space requirements. The effect in the performance is more obvious in BOOKS because of the long intervals and the high replication ratio. In view of these results, HINT$^m$ employs all optimizations from Section~\ref{sec:opts:comp} for the rest of our experiments.

\subsubsection{Handling data skewness \& sparsity and reducing cache misses}
\begin{table}[t]
\footnotesize
\caption{Optimizing HINT: impact of the skewness \& sparsity optimization (Section~\ref{sec:opts:skew}), default parameters}
\vspace{-2mm}
\begin{tabular}{|c||c|c|c|c|}
\hline
\multirow{2}{*}{\textbf{dataset}}	&\multicolumn{2}{c|}{\textbf{throughput} [queries/sec]}		&\multicolumn{2}{c|}{\textbf{index size} [MBs]}\\\cline{2-5}
													&original		&optimized										&original		&optimized\\
\hline\hline
BOOKS											&12098			&36173												&3282			&273\\
WEBKIT										&947			&39000											&49439		&337\\
TAXIS											&2931			&31027												&10093			&7733\\
GREEND										&648			&47038											&57667		&10131\\
\hline
\end{tabular}
\label{tab:hier-ds}
\vspace*{-3ex}
\end{table}
Table \ref{tab:hier-ds} tests the
effect of the {\em handling data skewness \& sparsity} optimization (Section~\ref{sec:opts:skew})
on the comparison-free version of HINT (Section~\ref{sec:hierarchical:precise}).%
\footnote{The {\em cache misses} optimization (Section \ref{sec:opts:cache}) is
only applicable to HINT$^m$.}
Observe that the optimization has a great effect on both the
throughput and the size of the index in all four real datasets, 
because empty partitions are effectively excluded from query evaluation and from the indexing process.

\begin{figure}[t]
\eat{
\begin{small}
\fbox{\parbox{190pt}
{
\begin{center}
{\footnotesize subs+sort+sopt}
\includegraphics[width=0.06\columnwidth]{figures/lines_2_2.pdf}
\hspace{3ex}
{\footnotesize skewness \& sparsity optimization}
\includegraphics[width=0.06\columnwidth]{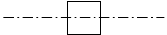}
\\
{\footnotesize cache misses optimization}
\includegraphics[width=0.06\columnwidth]{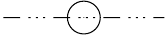}
\hspace{3.7ex}
{\footnotesize both optimizations}
\includegraphics[width=0.06\columnwidth]{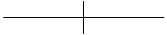}
\end{center}
}
}
\end{small}
}
\begin{tabular}{cc}
BOOKS &TAXIS\\
\hspace{-1ex}\includegraphics[width=0.46\columnwidth]{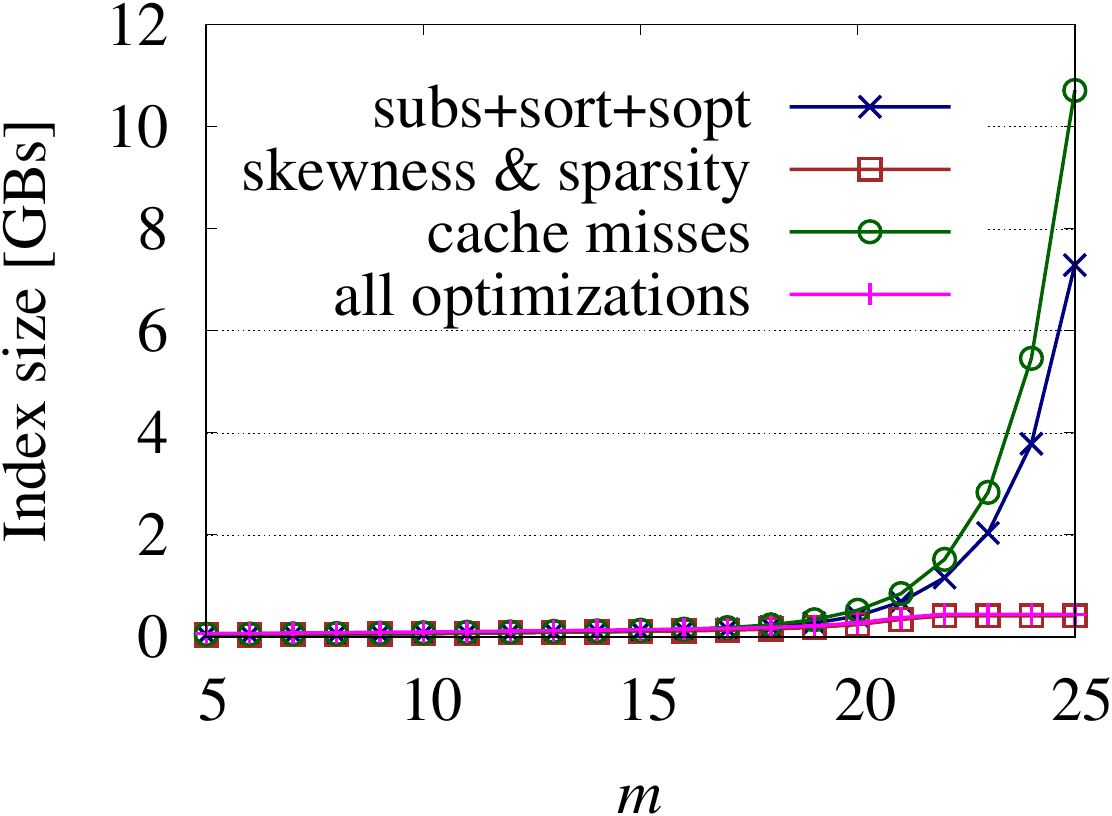}
&\includegraphics[width=0.46\columnwidth]{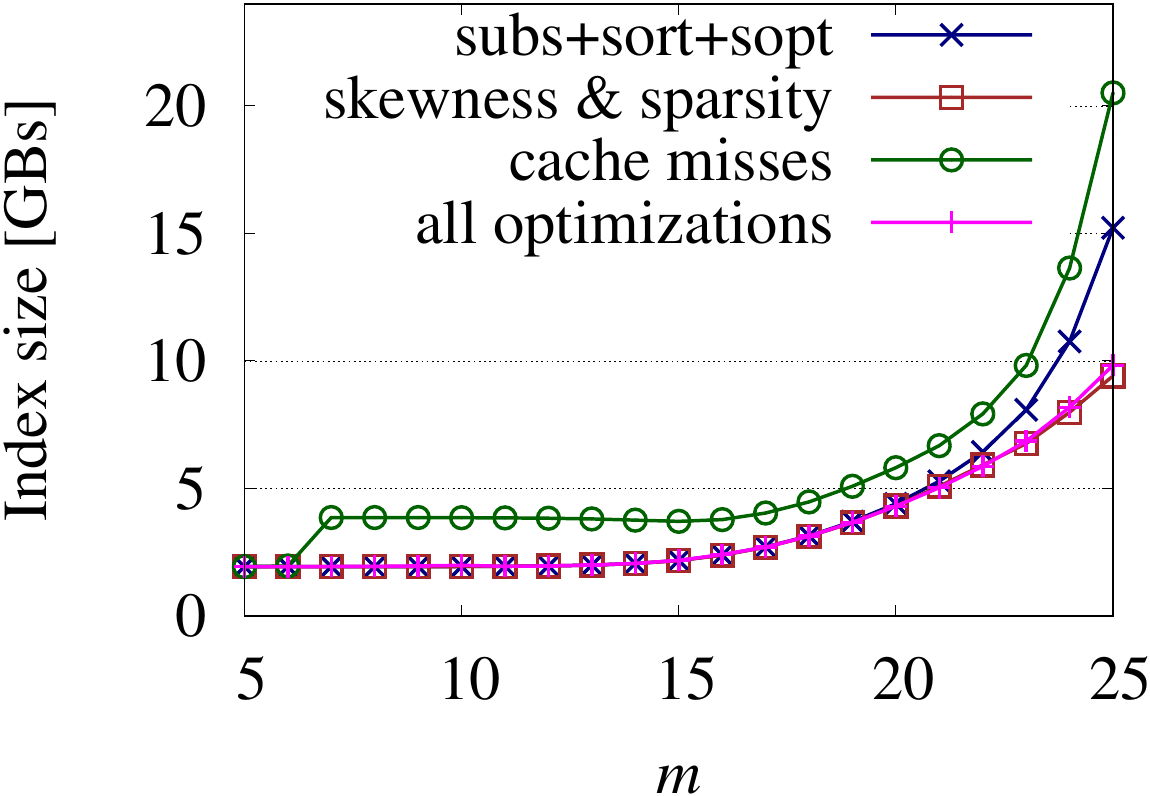}\\
\hspace{-1ex}\includegraphics[width=0.46\columnwidth]{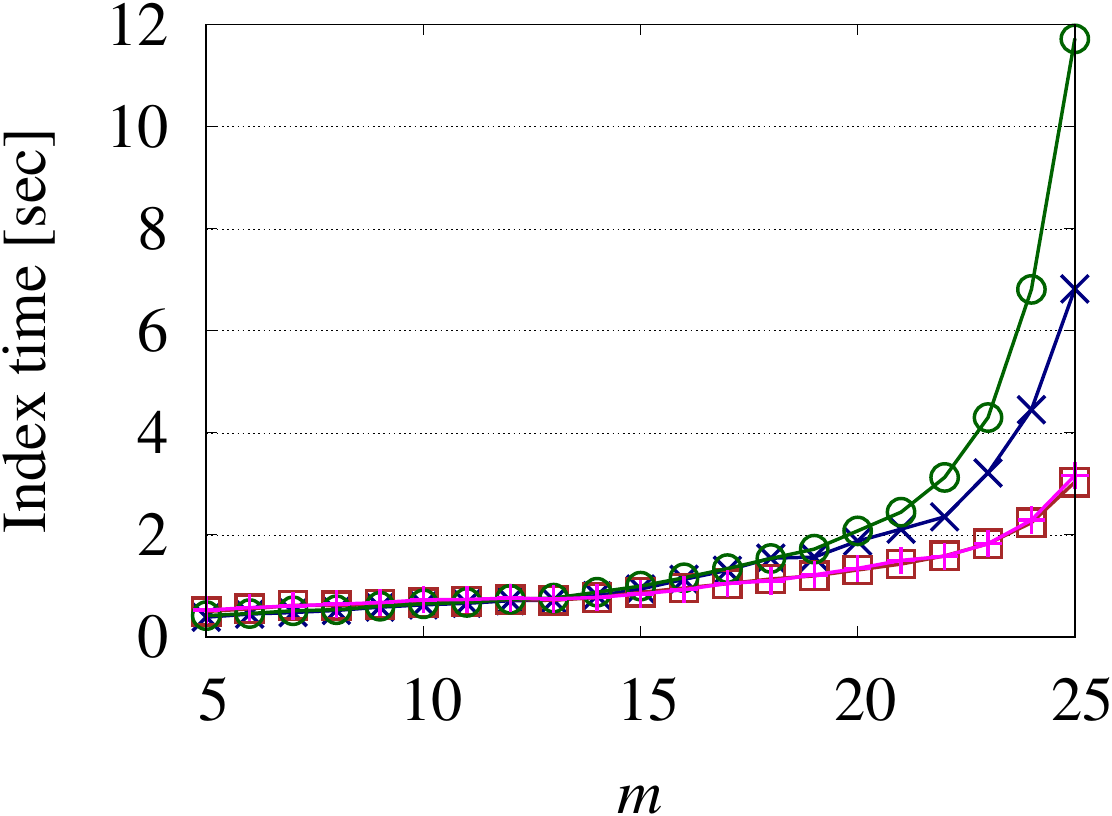}
&\includegraphics[width=0.46\columnwidth]{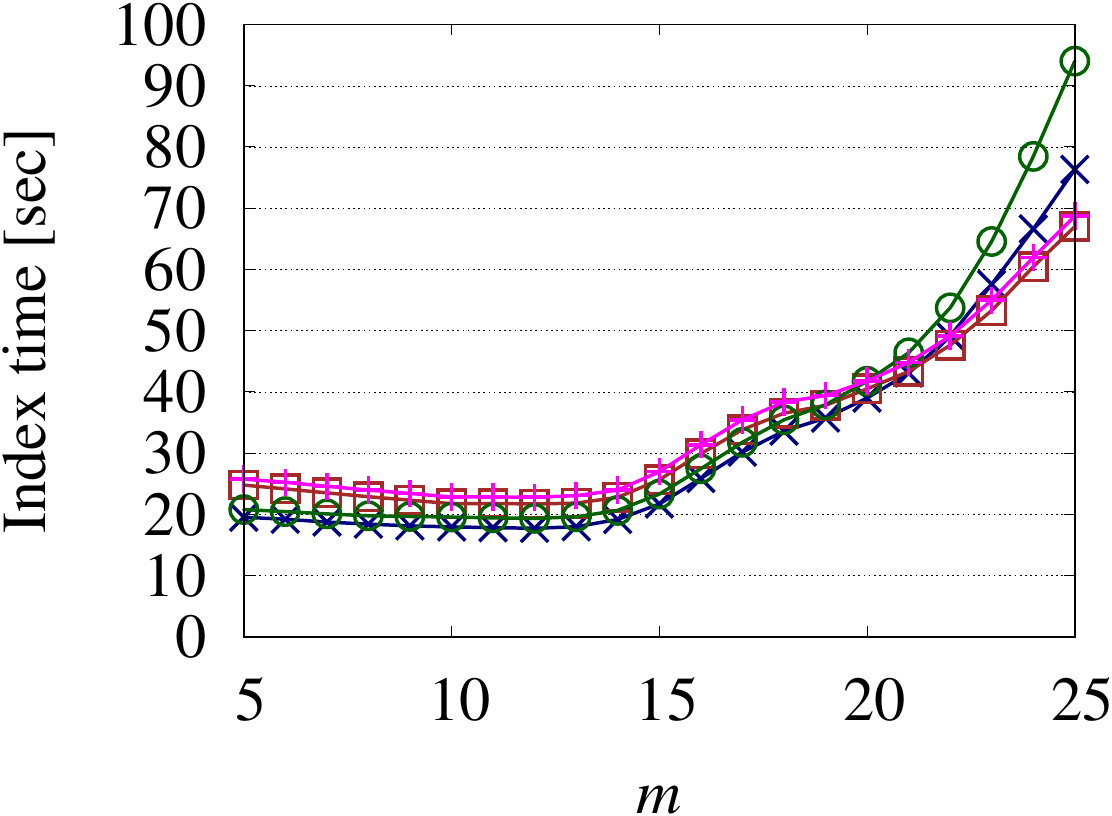}\\
\hspace{-1ex}\includegraphics[width=0.46\columnwidth]{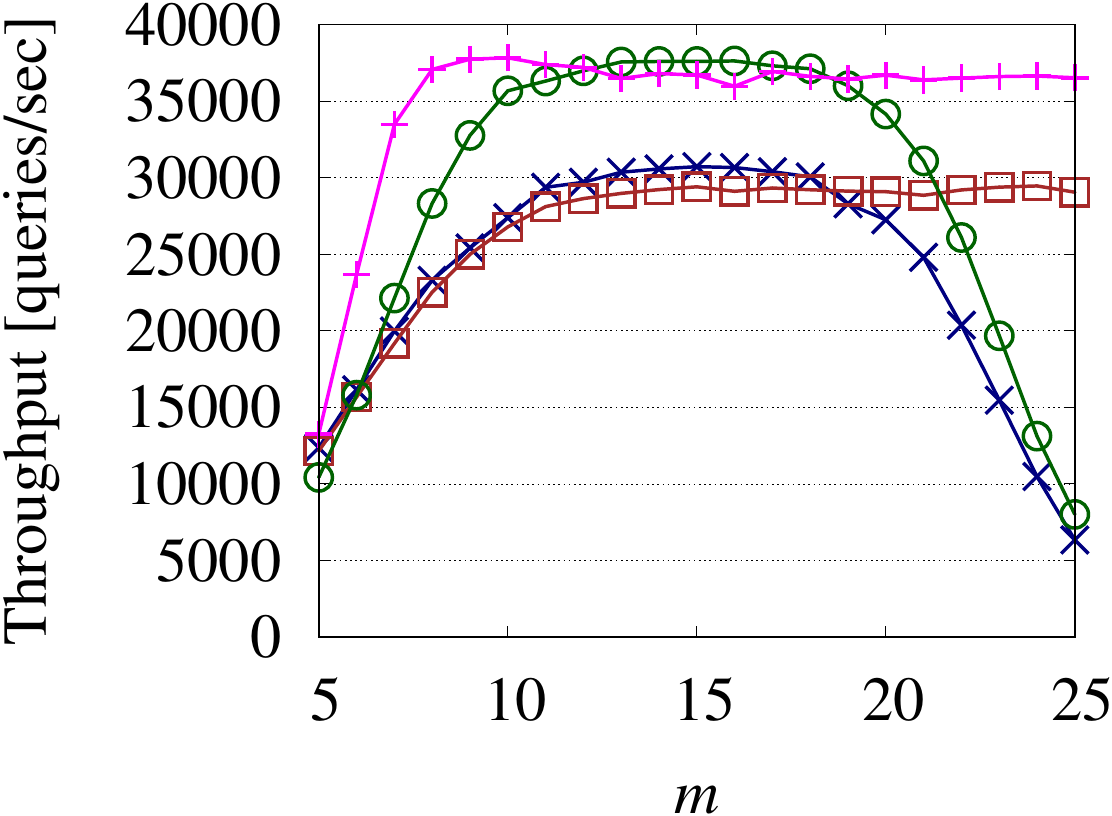}
&\includegraphics[width=0.46\columnwidth]{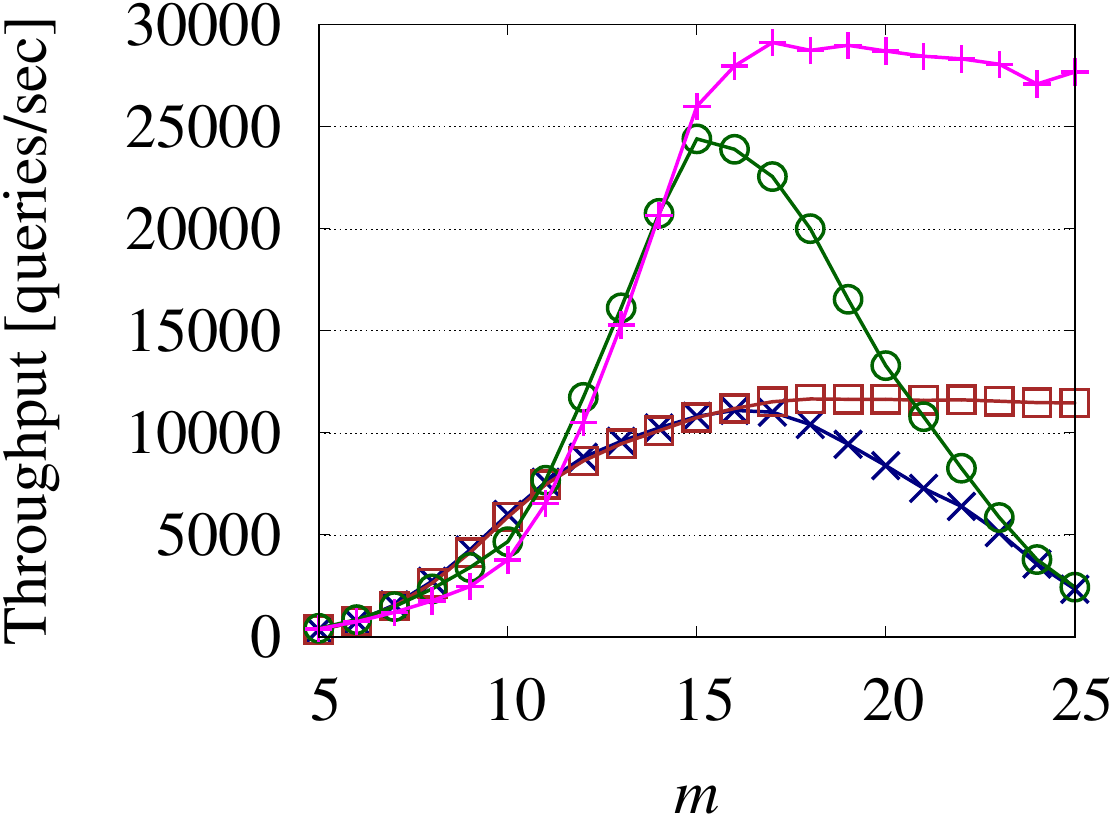}
\end{tabular}
\vspace{-3ex}
\caption{Optimizing HINT$^m$: impact of handling skewness \& sparsity and reducing cache misses optimizations}
\label{fig:hierm_rest}
\vspace*{-2ex}
\end{figure}
Figure \ref{fig:hierm_rest} shows the effect
of either or both of the \emph{data skewness \& sparsity} (Section~\ref{sec:opts:skew}) and the
\emph{cache misses} optimizations
(Section~\ref{sec:opts:cache})
on the
performance of HINT$^m$ for different values of $m$\eat{(i.e., the number of index levels used)}. 
In all cases\eat{for all values of $m$}, the
version of HINT$^m$ which uses both optimizations is superior
to all other versions. As expected, the {\em skewness \& sparsity} optimization helps
to reduce the space requirements of the index when $m$ is large,
because there are many empty partitions in this case at the bottom
levels of the index. At the same time, the \emph{cache misses} optimization
helps in reducing the number of cache misses in all cases where no
comparisons are needed. Overall, the optimized version of $HINT^m$
converges to its best performance at a relatively small value of $m$, where 
the space requirements of the index are relatively low, especially on the BOOKS and
WEBKIT datasets which contain long intervals.
For the rest of our experiments, HINT$^m$ employs both optimizations and HINT the \emph{data skewness \& sparsity} optimization.

\subsubsection{Discussion}
\begin{table}[t]
\footnotesize
\caption{Statistics and parameter setting}
\vspace{-2mm}
\begin{tabular}{|c|c||c|c|c|c|}
\hline
\textbf{index}							&\textbf{parameter}			&BOOKS		&WEBKIT		&TAXIS			&GREEND\\
\hline\hline
\multirow{2}{*}{Period}		&\# levels							&4				&4				&7				&8\\
												&\# coarse partitions			&100				&100				&100				&100\\  
												\hline 
Timeline							&\# checkpoints				&6000			&6000			&8000			&8000\\
\hline 
1D-grid										&\# partitions					&500			&300			&4000			&30000\\
\hline
\multirow{5}{*}{HINT$^m$} 		&\rev{$m_{opt}$ (model)}	&\rev{9}		&\rev{9}		&\rev{16}		&\rev{16}\\
												&$m_{opt}$ (exps)			&$10$			&$12$			&$17$			&$17$\\\cline{2-6}
												&\rev{rep. factor $k$ (model)}	&\rev{$6.09$}	&\rev{$8.98$}	&\rev{$1.98$}	&\rev{$1$}	\\
												&\rev{rep. factor $k$ (exps)}	&\rev{$5.13$}		&\rev{$6.07$}		&\rev{$2.14$}		&\rev{$1.0013$}\\\cline{2-6}
												&avg. comp. part.				&$3.226$		&$3.538$		&$3.856$		&$2.937$\\
\hline
\end{tabular}
\label{tab:stats}
\vspace*{-2ex}
\end{table}
\rev{
Table~\ref{tab:stats} reports the best values for parameter $m$ of HINT$^m$, denoted by $m_{opt}$.
For each real dataset, 
we show (1) $m_{opt}$ (model), estimated by our model in Section~\ref{sec:hierarchical:m} as the smallest $m$ value for which the index converges within 3\% to its lowest estimated cost, and (2) $m_{opt}$ (exps), which brings the highest throughput in our tests.
Overall, our model estimates a value of $m_{opt}$ which is very close
to the experimentally best value of $m$.
Despite a larger gap for WEBKIT,
the  measured throughput for the estimated \eat{value }$m_{opt}= 9$
is only 5\% lower than\eat{compared to} the  best observed throughput.
Further, the table shows the {\em replication factor}  $k$ of the index, i.e.,
the average number of partitions in which each interval is stored, as
predicted by our space complexity analysis 
(see Theorem~\ref{lem:spacecomp}) 
and as measured experimentally.
As expected, the replication factor is high on BOOKS, WEBKIT due to
the large number of long intervals, and low on TAXIS, GREEND where
the intervals are very short and stored at the bottom levels. Although
our analysis uses\eat{is based on} simple statistics, the
predictions are quite accurate.}
Finally, the last line Table~\ref{tab:stats}\eat{ of the table} ({\em avg. comp. part.})
shows the average number of
HINT$^m$ partitions for which comparisons were\eat{have to be} applied.
\eat{All numbers are below 4, which is consistent with our 
analysis in Section~\ref{sec:hint:analysis}.
In practice, this}
Consistently to our analysis in Section~\ref{sec:hint:analysis}, all numbers are below 4, which
means that the performance of HINT$^m$ is very close
to the performance of the comparison-free, but space-demanding HINT.

\eat{
\begin{table}[t]
\footnotesize
\caption{Comparing index size [MBs]}
\vspace{-2mm}
\begin{tabular}{|c||c|c|c|c|}
\hline
\textbf{index}		&BOOKS	&WEBKIT	&TAXIS			&GREEND\\
\hline\hline
Interval tree			&97			&115			&3125			&2241\\
Period 			&210			&217			&2278			&\textbf{1262}\\
Timeline 		&4916		&5671		&4203			&2525\\
1D-grid						&949		&604		&\textbf{2165}			&1264\\
\hline
HINT						&273		&337		&7733			&10131\\
HINT$^m$			&\textbf{81}			&\textbf{97}			&3048			&1278\\
\hline
\end{tabular}
\vspace{-2ex}
\label{tab:indexsize}
\end{table}
}
\begin{table}[t]
\footnotesize
\caption{Comparing index size [MBs]}
\vspace{-2mm}
\begin{tabular}{|c||c|c|c|c|}
\hline
\textbf{index}		&BOOKS	&WEBKIT	&TAXIS			&GREEND\\
\hline\hline
Interval tree			&97			&115			&3125			&2241\\
Period 			&210			&217			&2278			&\textbf{1262}\\
Timeline 		&4916		&5671		&4203			&2525\\
1D-grid					&949		&604		&2165			&1264\\
\hline
HINT						&273		&337		&7733			&10131\\
HINT$^m$			&\textbf{81}			&\textbf{98}			&\textbf{2039}			&1278\\
\hline
\end{tabular}
\vspace{-2ex}
\label{tab:indexsize}
\end{table}
\eat{
\begin{table}[t]
\footnotesize
\caption{Comparing index time [sec]}
\vspace{-2mm}
\begin{tabular}{|c||c|c|c|c|}
\hline
\textbf{index}	&BOOKS		&WEBKIT		&TAXIS			&GREEND\\
\hline\hline
Interval tree		&\textbf{0.249647}	&\textbf{0.333642}	&47.1913		&26.8279\\
Period 		&1.14919		&1.35353		&76.9302		&46.3992\\
Timeline 	&12.665271	&19.242939	&40.376573	&15.962221\\
1D-grid				&1.26315		&0.952408	&\textbf{4.02325}		&\textbf{2.23768}\\
\hline
HINT					&1.70093		&11.7671		&49.589		&36.5143\\
HINT$^m$		&1.15779		&1.15857		&39.4318		&8.39652\\
\hline
\end{tabular}
\label{tab:indexbuildtime}
\vspace{-3ex}
\end{table}
}
\begin{table}[t]
\footnotesize
\caption{Comparing index time [sec]}
\vspace{-2mm}
\begin{tabular}{|c||c|c|c|c|}
\hline
\textbf{index}	&BOOKS		&WEBKIT		&TAXIS			&GREEND\\
\hline\hline
Interval tree		&\textbf{0.249647}	&\textbf{0.333642}	&47.1913		&26.8279\\
Period 		&1.14919		&1.35353		&76.9302		&46.3992\\
Timeline 	&12.665271	&19.242939	&40.376573	&15.962221\\
1D-grid				&1.26315		&0.952408	&\textbf{4.02325}		&\textbf{2.23768}\\
\hline
HINT					&1.70093		&11.7671		&49.589		&36.5143\\
HINT$^m$		&0.725174		&0.525927		&22.787983		&8.577486\\
\hline
\end{tabular}
\label{tab:indexbuildtime}
\vspace{-3ex}
\end{table}

\subsection{Index performance comparison}

\begin{figure*}[t]
\begin{small}
\fbox{\parbox{325pt}
{
\begin{center}
{\footnotesize Interval tree}
\includegraphics[width=0.06\columnwidth]{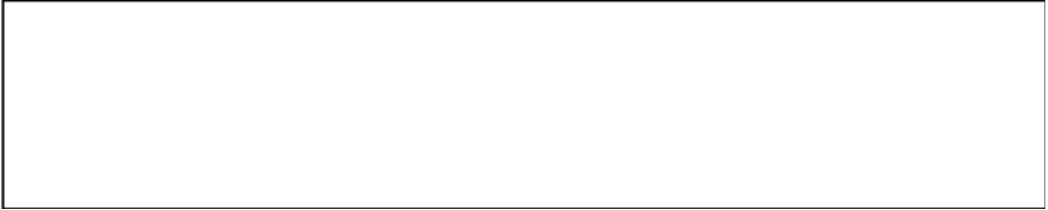}
\hspace{2ex}
{\footnotesize Period index}
\includegraphics[width=0.06\columnwidth]{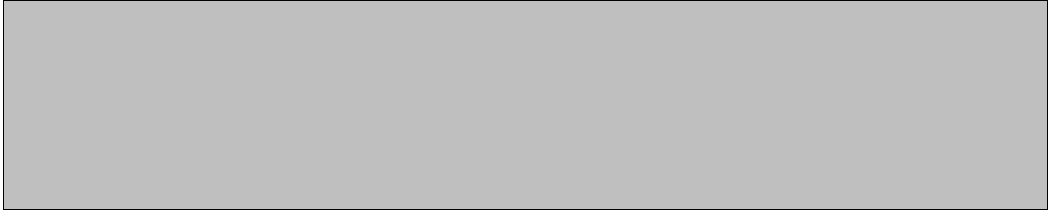}
\hspace{2ex}
{\footnotesize Timeline index}
\includegraphics[width=0.06\columnwidth]{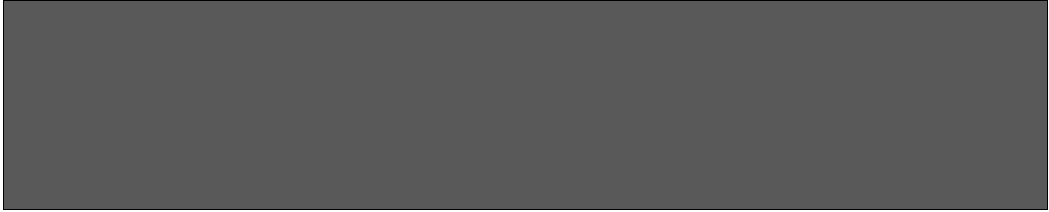}
\hspace{2ex}
{\footnotesize 1D-grid}
\includegraphics[width=0.06\columnwidth]{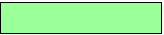}
\hspace{2ex}
{\footnotesize 	HINT}
\includegraphics[width=0.06\columnwidth]{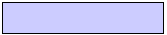}
\hspace{2ex}
{\footnotesize HINT$^m$}
\includegraphics[width=0.06\columnwidth]{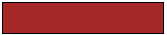}
\end{center}
}
}
\end{small}
\begin{tabular}{cccc}
BOOKS &WEBKIT &TAXIS &GREEND\\
\hspace{-1ex}\includegraphics[width=0.5\columnwidth]{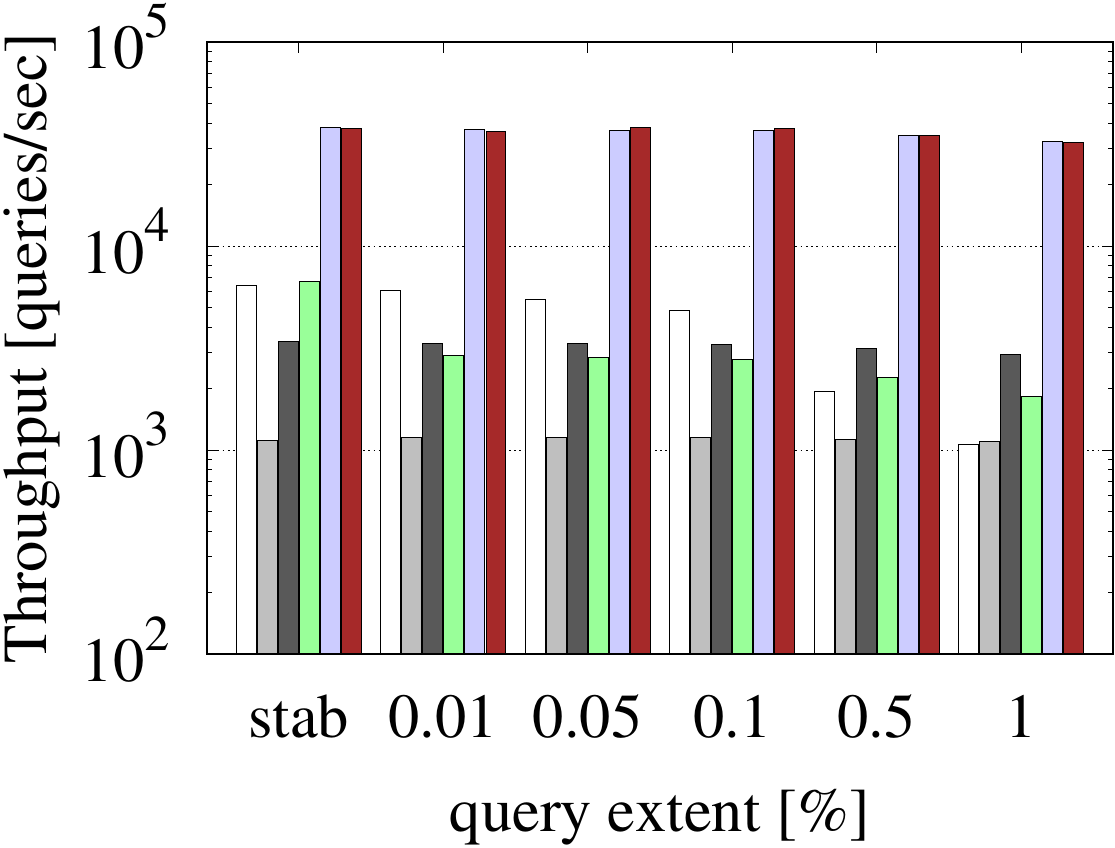}
&\hspace{-1ex}\includegraphics[width=0.5\columnwidth]{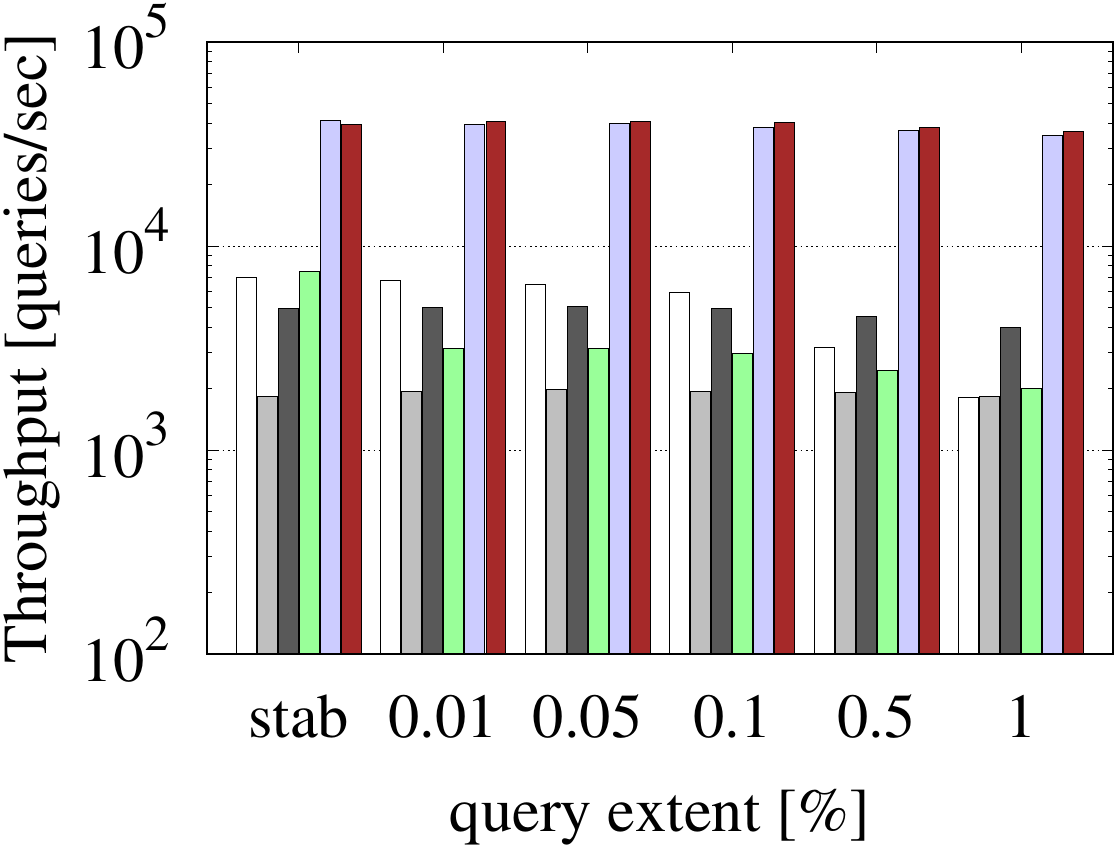}
&\hspace{-1ex}\includegraphics[width=0.5\columnwidth]{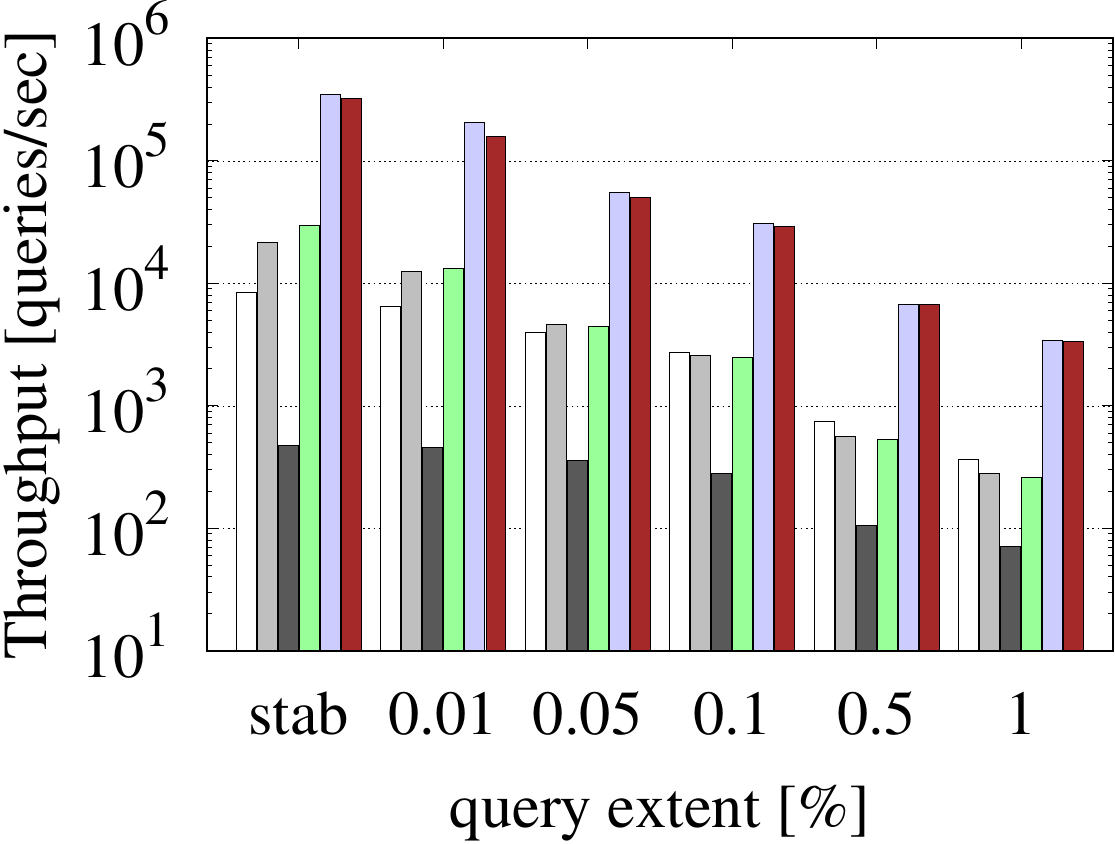}
&\hspace{-1ex}\includegraphics[width=0.5\columnwidth]{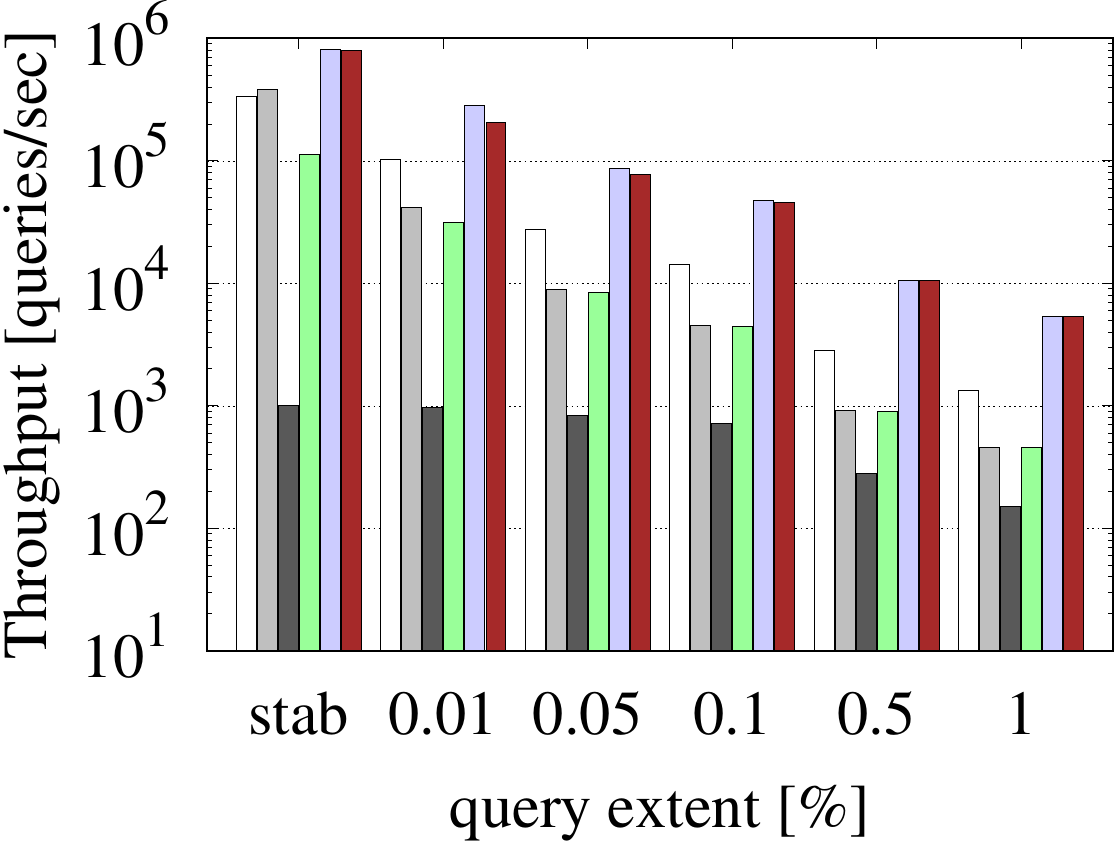}
\end{tabular}
\vspace{-3ex}
\caption{Comparing throughputs, real datasets}
\vspace{-2ex}
\label{fig:comparison}
\end{figure*}
\begin{figure*}[ht]
\begin{tabular}{ccccc}
\hspace{-1ex}\includegraphics[width=0.4\columnwidth]{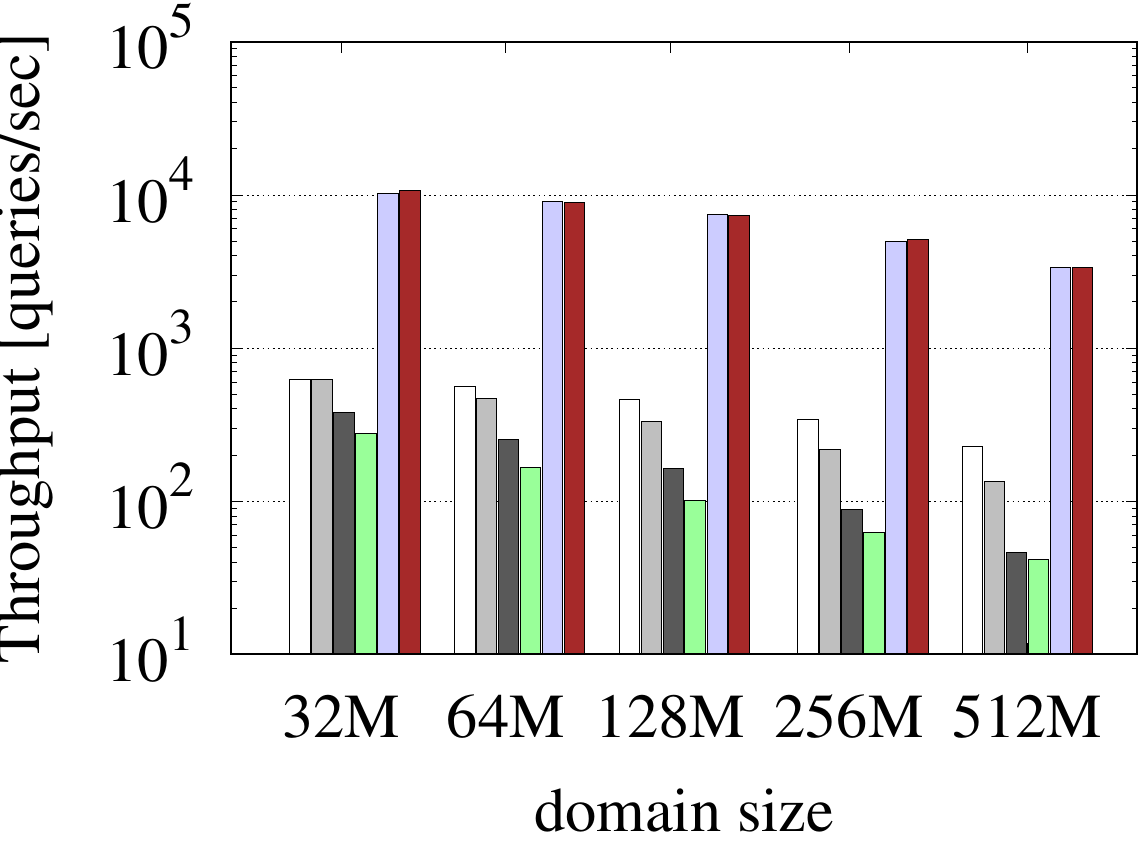}
&\hspace{-1ex}\includegraphics[width=0.4\columnwidth]{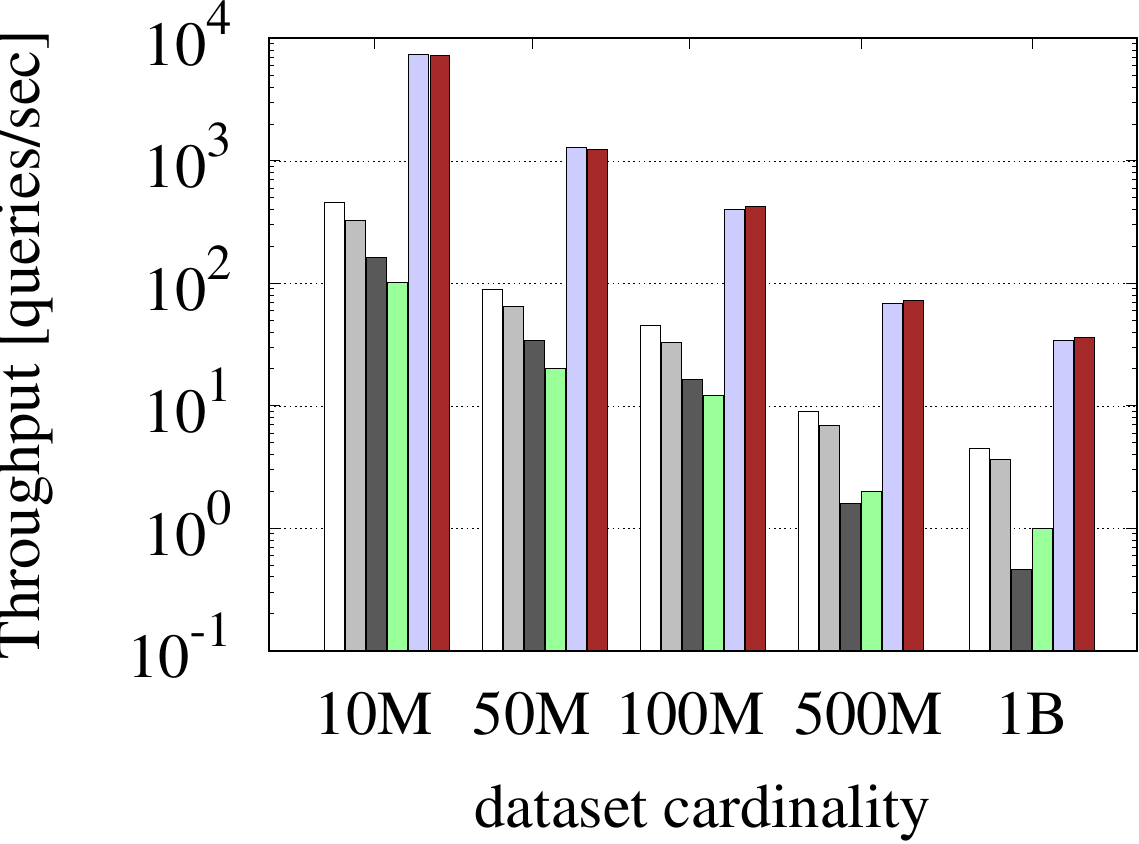}
&\hspace{-1ex}\includegraphics[width=0.4\columnwidth]{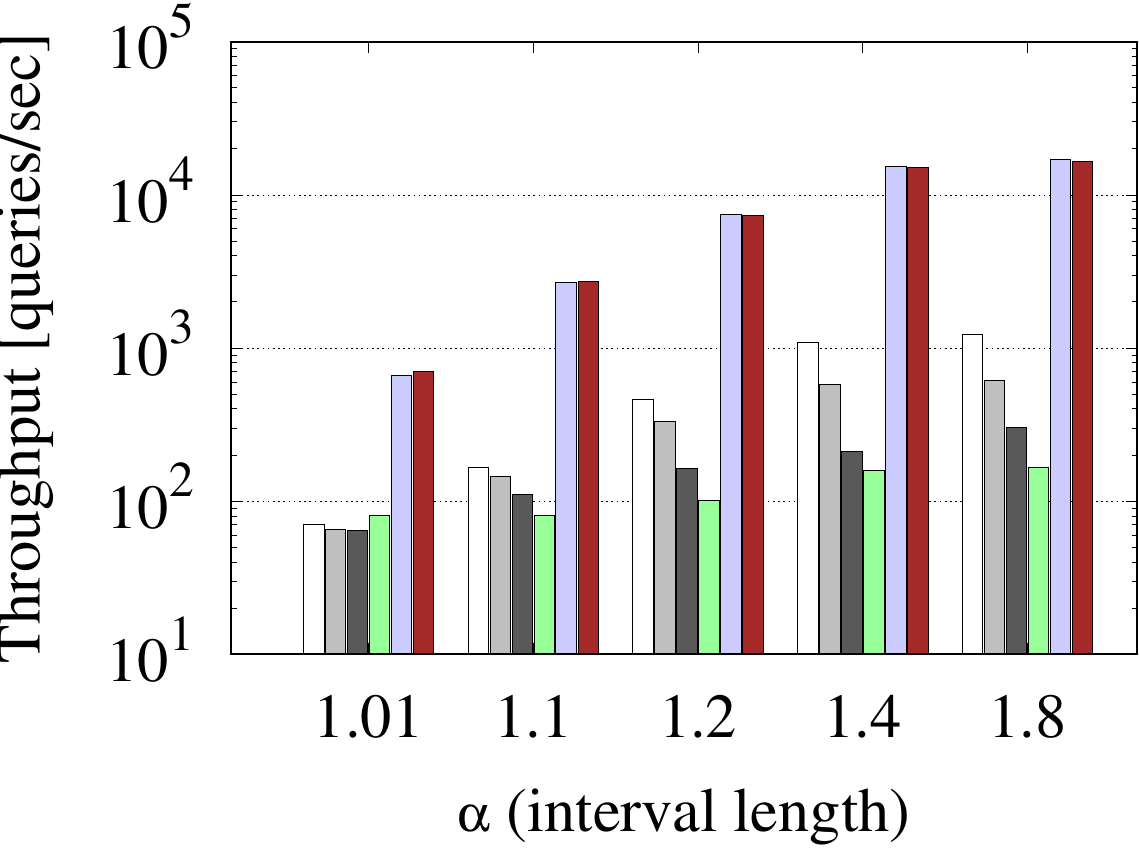}
&\hspace{-1ex}\includegraphics[width=0.4\columnwidth]{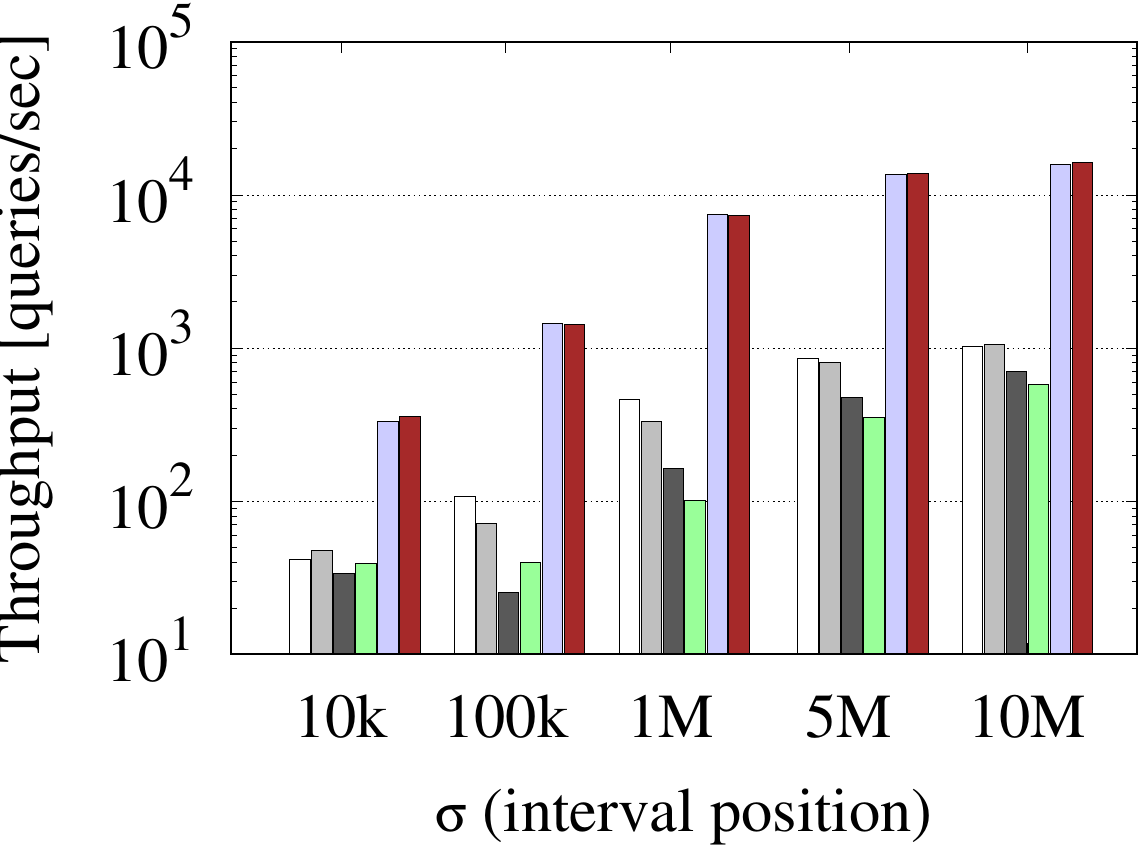}
&\hspace{-1ex}\includegraphics[width=0.4\columnwidth]{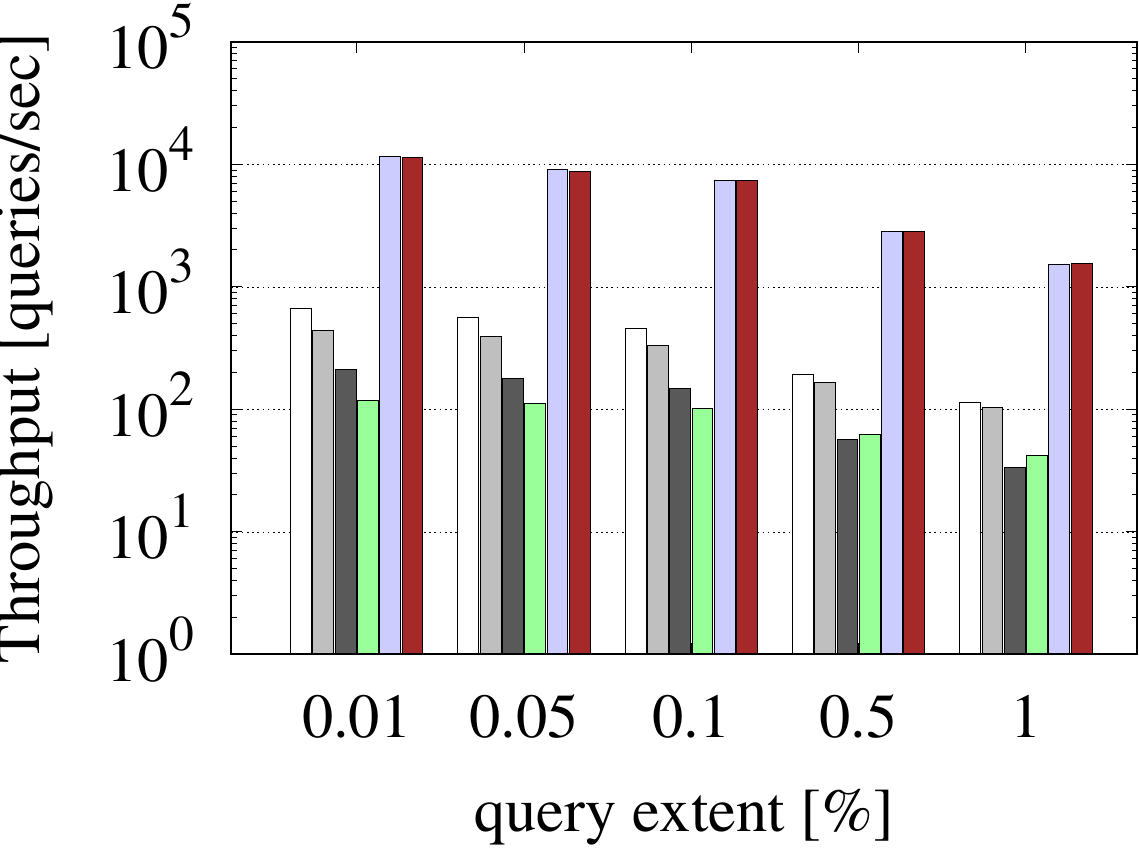}
\end{tabular}
\vspace{-3ex}
\caption{Comparing throughputs, synthetic datasets}
\vspace{-2ex}
\label{fig:comparison_synth}
\end{figure*}
We next compare the optimized versions of HINT and HINT$^m$
against the previous work competitors.
We start with our tests on the real datasets.
For HINT$^m$, we set $m$ to the best value on each dataset, according to
Table \ref{tab:stats}. Similarly, we set the number of partitions for 1D-grid, the number of checkpoints for the timeline index, and the number of levels and number of coarse partitions for the period index (see Table~\ref{tab:stats}).
Table \ref{tab:indexsize} shows the sizes of each index
in memory and Table \ref{tab:indexbuildtime} shows the construction
cost of each index, for the default query extent 0.1\%.
Regarding space, HINT$^m$ along with the interval tree and the period
index have the lowest requirements on datasets with long intervals
(BOOKS and WEBKIT) and very similar to 1D-grid in the rest.
In TAXIS and GREEND where the intervals are indexed mainly at the
bottom level, the space requirements of HINT$^m$ are significantly
lower than our comparison-free HINT due to limiting the number of levels.
\rev{When compared to the raw data (see Table \ref{tab:datasets}),
  HINT$^m$ is 2 to 3 times bigger for BOOKS and WEBKIT (which contain
  many long intervals), and 1\eat{to 2} time bigger for GREEND and TAXIS.
These ratios are smaller than the replication ratios $k$ reported in
Table \ref{tab:stats},
due to our storage optimization (cf. Section \ref{sec:flat:plus:decomposition}).}
 Due to its simplicity, 1D-grid has the lowest index time
across all datasets. Nevertheless, HINT$^m$ is the runner up in most
of the cases, especially for the biggest inputs, i.e., TAXIS and
GREEND, while in BOOKS and WEBKIT, its index time is very 
close to the interval tree.

Figure \ref{fig:comparison} compares the query throughputs of all
indices on queries of various extents (as a percentage of the domain
size). The first set of bars in each plot corresponds to {\em stabbing}
queries, i.e., range queries of 0 extent.
We observe that HINT and HINT$^m$ outperform the competition
by almost one order of magnitude, across the board. In fact, only on GREEND the 
performance for one of the competitors, i.e., 1D-grid, comes close to the performance 
of our hierarchical indexing. Due to the extremely short intervals \eat{contained }in GREEND (see Table~\ref{tab:datasets}) the vast majority of the results are collected from the bottom level of HINT/HINT$^m$, which essentially resembles the evaluation process in 1D-grid. Nevertheless, our \eat{hierarchical }indices are even in this case faster as they require no duplicate elimination. 

\eat{FLAT+ achieves very good performance as well. 
In fact for very short queries, FLAT+ may outperform HINT and
HINT$^m$, as in these cases the query range is usually completely
contained inside a single FLAT+ partition.
On the other hand, FLAT+ requires
significantly more space for datasets with long intervals (BOOKS and
WEBKIT) (see Table~\ref{tab:statsFLATHINT}). 
}

HINT$^m$ is the best index overall, as it achieves the performance of
HINT, requiring less space, confirming 
the findings of our analysis in Section~\ref{sec:hint:analysis}.
As shown in Table~\ref{tab:indexsize}, HINT always has higher space
requirements than HINT$^m$; even up to an order of magnitude higher in
case of GREEND.
What is more, since HINT$^m$ offers the option to control the occupied
space in memory by appropriately setting the $m$ parameter,
it can handle scenarios with space limitations.
HINT is marginally better than HINT$^m$ only on datasets with short
intervals (TAXIS and GREEND) and only for selective queries. 
In these cases, the intervals are stored at the lowest levels of the
hierarchy where HINT$^m$ typically needs to conduct comparisons
to identify results, but HINT applies comparison-free retrieval.

The next set of tests are on synthetic datasets. 
In each test, we fix all but one parameters (domain size,
cardinality, $\alpha$, $\sigma$, query extent) to their default values and
varied one (see Table~\ref{tab:synthdatasets}).
The value of $m$ for HINT$^m$, the number of partitions for 1D-grid, the number of checkpoints for the timeline index and the number of levels/coarse partitions for the period index are set to their
best values on each dataset. 
The results, shown in  Figure~\ref{fig:comparison_synth}, follow a
similar trend to the tests on the real datasets. 
HINT and HINT$^m$ are
always significantly faster than the competition, .
Different to the \eat{case of the }real datasets, 1D-grid is steadily outperformed by the other three competitors. Intuitively, the uniform partitioning of the domain \eat{employed by}in 1D-grid cannot cope with the skewness of the synthetic datasets\eat{; in contrast, HINT and HINT$^m$ use a much finer partitioning}.
As expected the domain size, the dataset cardinality and the query extent have a negative impact on the performance of all indices. 
Essentially, increasing the domain size under a fixed query extent,
affects the performance similar to increasing the query extent, i.e.,
the queries become longer and less selective, including more
results. Further, the querying cost grows linearly with the dataset
size since the number of query results are proportional to it.
\rev{HINT$^m$ occupies around 8\% more space than the raw data,
  because the replication factor $k$ is close to 1.}
In contrast, as $\alpha$ grows, the intervals become shorter, so the query performance improves.
Similarly, when increasing $\sigma$ the intervals are more widespread,
\eat{which means}meaning that the queries are expected to retrieve fewer results,
and the query cost drops accordingly.

\begin{table}[t]
\centering
\caption{\rev{Throughput [operations/sec] and total cost [sec]}}
\vspace{-2ex}
\label{tab:updates}
BOOKS\\\vspace{1ex}
\footnotesize
\begin{tabular}{|@{~}c@{~}||@{~}c@{~}|c@{~}|c@{~}|@{~}c@{~}|c@{~}|}
\hline
\textbf{operation}		&Interval tree	&Period index	&1D-grid		&$_\mathrm{subs+sopt}$HINT$^m$		&HINT$^m$\\\hline\hline
queries						&1,258		&3,088			&3,739			&14,390									&{\bf40,311}\\
insertions					&5,841		&519,904		&411,540		&2,405,228								&{\bf3,680,457}\\
deletions					&1,142		&765			&165			&2,201									&{\bf5,928}\\
\hline\hline
\textbf{total cost}		&9.63			&4.52			&8.68			&1.14									&{\bf0.41}\\
\hline
\end{tabular}
\\\vspace{2ex}
\normalsize
TAXIS\\\vspace{1ex}
\footnotesize
\begin{tabular}{|@{~}c@{~}||@{~}c@{~}|c@{~}|@{~}c@{~}|@{~}c@{~}|c@{~}|}
\hline
\textbf{operation}		&Interval tree	&Period index	&1D-grid			&$_\mathrm{subs+sopt}$HINT$^m$		&HINT$^m$\\\hline\hline
queries						&2,619		&2,695			&2,572				&8,774								&{\bf28,596}\\
insertions					&61,923		&1,026,423		&{\bf 8,347,273}	&4,407,743							&6,745,622\\	
deletions					&14,318		&21,293			&16,236				&71,122								&{\bf90,460}\\
\hline\hline
\textbf{total cost}		&3.93			&3.76				&3.95			&1.15								&{\bf0.36}\\
\hline
\end{tabular}
\vspace{-2ex}
\end{table}

\subsection{Updates}
\rev{
Finally, we test the efficiency of HINT$^m$ in updates using
both the update-friendly version of HINT$^m$ (Section
\ref{sec:hierarchical:updates}),
denoted by $_\mathrm{subs+sopt}$HINT$^m$,
and the hybrid setting for the fully-optimized index from
Section~\ref{sec:opts:updates}, denoted as HINT$^m$. 
\eat{Specifically, w}We
index offline the first 90\% of the intervals for each real dataset in
batch and then execute a mixed workload with 10K range queries of
0.1\% extent, 5K insertions of new intervals (randomly selected from
the remaining 10\% of the dataset) and 1K random deletions.
Table~\ref{tab:updates} reports our findings for BOOKS and TAXIS; the
results for WEBKIT and GREEND follow the same trend. Note that we
excluded Timeline\eat{ from our analysis} since the index is
designed for temporal (versioned) data where updates only happen as
new events are appended at the end of the event list, and the
comparison-free HINT, for which our tests have already shown a
similar performance to HINT$^m$ with higher indexing/storing
costs. 
Also, all indices handle deletions with ``tombstones''. 
We observe that both versions of HINT$^m$ outperform the competition by a
wide margin. 
An exception arises on TAXIS, as the short intervals are inserted in only one partition in 1D-grid.
The interval tree has in fact several orders of magnitude
slower updates due to the extra cost of maintaining the partitions in
the tree sorted at all time. Overall, we also observe that the hybrid 
%
HINT$^m$ setting is the most efficient index as the \emph{smaller} delta
$_\mathrm{subs+sopt}$HINT$^m$ handles insertions faster
than the 90\% pre-filled $_\mathrm{subs+sopt}$HINT$^m$.
%
}


\section{Conclusions and Future Work}
\label{sec:concl}

We proposed a hierarchical index (HINT) for intervals, which has
low space complexity and minimizes the number of data accesses and
comparisons during query evaluation.
Our experiment\eat{al analysi}s on real and synthetic datasets
shows that HINT outperforms previous work by one order of magnitude in a wide
variety of interval data and query distributions.
There are several directions for future work. 
First, we plan to study
the performance of HINT on selection queries,
based on Allen's relationships \cite{Allen81} between intervals \rev{and on
complex event processing in data streams, based on interval
operators \cite{Awad0KVS20}}.
Second, we plan to investigate extensions of HINT for supporting
queries that combine temporal selections and selections on
additional object attributes or the duration of
intervals \cite{BehrendDGSVRK19}.
Third, we plan to investigate effective parallelization
techniques,
taking advantage of the fact that HINT partitions are independent.

\eat{

\todo{Application to range temporal aggregates \cite{ZhangMTGS01} and
  interval joins.}

\todo{Other queries: (1) Allen's, (2) selection on multiple attributes
(besides interval-based), (3) duration}

\todo{Adaptivity: (1) unbalanced partitioning (2) skip levels}

\todo{Dynamic environments and updates}
\begin{itemize}
	\item append-only updates (streaming, temporal)
	\item append-front, delete-tail updates (sliding window)
	\item ad-hoc updates
\end{itemize}
(some discussions on this could be included in the paper)

\todo{Parallel Processing} query evaluation at each (sub)partition and
at each level of HINT is independent to other partitions, ....

\todo{adaptivity}
\begin{itemize}
	\item uneven depth
	\item jump levels
	\item adaptive partitioning per level (BOOKS)
        \end{itemize}

\todo{time-series segmentation and indexing - new paper}
}

\balance
\bibliographystyle{ACM-Reference-Format}
\bibliography{iindex}

\end{document}